%% file: jsac.tex
\documentclass[10pt,twocolumn,twoside]{IEEEtran}
\makeatletter
\def\ps@headings{%
\def\@oddhead{\mbox{}\scriptsize\rightmark \hfil \thepage}%
\def\@evenhead{\scriptsize\thepage \hfil \leftmark\mbox{}}%
\def\@oddfoot{}%
\def\@evenfoot{}}
\makeatother
\input{Commands}

\makeatletter
\renewcommand\paragraph{\@startsection{paragraph}{4}{\z@}%
    {1.5ex plus .2ex minus .3ex}%
            {-0em}%
                        {\normalsize\bf}}

\newcommand{\captionfonts}{\small}

\long\def\@makecaption#1#2{%
  \vskip\abovecaptionskip
  \sbox\@tempboxa{{\captionfonts #1: #2}}%
  \ifdim \wd\@tempboxa >\hsize
    {\captionfonts #1: #2\par}
  \else
    \hbox to\hsize{\hfil\box\@tempboxa\hfil}%
  \fi
  \vskip\belowcaptionskip}
\makeatother

\begin{document}
\title{On Budgeted Influence Maximization\\ in Social Networks
}
\author{
\begin{tabular}{ccc}
Huy Nguyen & \hspace{1in} & Rong Zheng  \\
Department of Computer Science & \hspace{1in} & Department of Computing and Software \\
University of Houston & \hspace{1in} & McMaster University \\
Houston, TX 77204 & \hspace{1in} & Hamilton, ON, L8S 4K1, Canada\\
{\it hanguyen5@uh.edu} & \hspace{1in} & {\it rzheng@mcmaster.ca}
\vspace{-0.2in}
\end{tabular}
}
\maketitle
\thispagestyle{empty}
\pagestyle{empty}
\begin{abstract}
Given a budget and arbitrary cost for selecting each node, the budgeted influence maximization (BIM) problem concerns selecting a set of seed nodes to disseminate some information that maximizes the total number of nodes influenced (termed as influence spread) in social networks at a total cost no more than the budget.
Our proposed seed selection algorithm for the BIM problem guarantees an approximation ratio of $(1-1/\sqrt{e})$. The seed selection algorithm needs to calculate the influence spread of candidate seed sets, which is known to be \#P-complex. Identifying the linkage between the computation of marginal probabilities in Bayesian networks and the influence spread, we devise efficient heuristic algorithms for the latter problem.
Experiments using both large-scale social networks and synthetically generated networks demonstrate superior performance of the proposed algorithm with moderate computation costs. Moreover, synthetic datasets allow us to vary the network parameters and gain important insights on the impact of graph structures on the performance of different algorithms.

\end{abstract}
\section{Introduction}
\label{sec:intro}
The social network of interactions among a group of individuals plays a fundamental role in the spread of information, ideas, and influence. Such effects have been observed in real life, when an idea or an action gains sudden widespread popularity through {\it ``word-of-mouth''} or {\it ``viral marketing''} effects. For example, free e-mail services such as Microsoft's Hotmail, later Google's Gmail, and most recently Google's Google+ achieved wide usage largely through referrals, rather than direct advertising.

In viral marketing, one important question is given limited advertisement resources, which set of customers should be targeted such that the resulting influenced population is maximized. Consider a social network modeled as a graph with vertices representing individuals and edges representing connections or relationship between two individuals.
The influence maximization (IM) problem tries to find a seed set $S$ with cardinality $|S| = k$ in the graph such that the expected number of nodes influenced by $S$ is maximized~\cite{Domingos01,Richardson02,kempe03}. With the cardinality constraint, the submodularity nature of the influence spread renders a greedy algorithm with $(1-1/e)$ approximate ratio that in each round picks the seed with maximum influence spread and runs for $k$ rounds.
However the assumption of equal costs for all seed nodes seldom holds in practice.
With the proliferation of influence score services such as Klout and PeerIndex\footnote{{\it http://www.klout.com} and {\it http://www.peerindex.com}}, one can easily measure his influence in the social sphere and use that to negotiate the price for services he provides. The higher the influence score of a user, the more costly it is to persuade him.

We consider in this paper a generalized version of the IM problem, namely, the budgeted influence maximization (BIM) problem: given a fixed budget $b$ and a random cost function $c$, find a seed set $S$ which fits the budget $\sum_{s_i \in S}c(s_i) \leq b$ and maximizes the number of influenced nodes. Clearly, BIM is more relevant in practice as there is typically a price associated with initializing the dissemination of information. With the budget constraint, we prove that direct application of the simple greedy algorithm may result in unbounded performance gap.

In this paper, we present a seed selection algorithm that can attain an approximation guarantee of $(1-1/\sqrt{e})$ ($\sim 0.394$).
One critical component of the seed selection process is the determination of influence spread of a set of seeds. Exact computation of influence spread is proven to be of \#P-complete~\cite{kempe03}. Thus, efficient algorithms need to be devised. More specifically, we first establish the linkage between influence spread computation and belief propagation on a Bayesian network (modeled as a directed acyclic graph [DAG]), where the marginal conditional dependency corresponds to the influence probabilities. Belief propagation has been extensively studied in literatures, and thus many exact or approximation algorithms can be leveraged to estimate the influence spread. For a general graph that contains loops, we propose two approximation algorithms that prune some edges in the graph to obtain a DAG that captures the bulk of influence spread. To reduce the number of candidate seed nodes, we localize the influence spread region such that at each round, only nodes that are affected by the previously selected seed need to be evaluated.
Empirical study shows that the proposed algorithms can scale up to large-scale graphs with millions of edges with high accuracy. On real-world social network graphs, our methods achieve influence spread comparable to that by Greedy algorithm~\cite{kempe03} and incur significant less computation costs. In the unit-cost IM problem, the proposed methods outperform PMIA~\cite{pmia} in achievable influence spread at the expense of marginal increase in computation time. In the BIM problem, the proposed methods outperform CELF~\cite{leskovec07} in term of scalability and performance on dense graphs. We further study the effect of network structures on the performance of the algorithms.

The main contributions of this paper are summarized as follows:
\begin{itemize}
\item We propose a greedy algorithm for BIM with a constant approximation ratio.
\item We cast the problem of inference spread computation on a DAG as an instance of belief propagation on a Bayesian network.
\item We prove the \#P-hardness of inference spread computation on a DAG.
\item Two heuristics are proposed to construct DAGs from a general graph that capture the bulk of influence spread.
\item We provide important insights on the impact of graph structures on performance of different algorithms.
\end{itemize}

The rest of this paper is organized as follows. In Section~\ref{sec:previouswork}, we give a comprehensive review of the related literatures. Section~\ref{sec:bim} presents the seed selection algorithm and proves its performance bound.
Theoretical results concerning influence spread on DAGs are in Section~\ref{sec:spreadondag}. In Section~\ref{sec:dagmodels}, we devise two heuristics to reduce a general directed graph into a DAG which captures the majority of influence spread. From the presented theoretical results, we have the main algorithm in Section~\ref{sec:proposedalgorithm}. In Section~\ref{sec:eval}, extensive experiment results are presented. Finally we conclude the paper and discuss future research directions in Section~\ref{sec:conclusion}.

\section{Related Work}
\label{sec:previouswork}
Kempe {\it et al.} in~\cite{kempe03} are the first to formulate the IM problem. The authors proved the submodularity of the influence spread function and suggested a greedy scheme (henceforth referred to as Greedy algorithm) with an incremental oracle that identifies, in each iteration, a new seed that maximizes the incremental spread. The approach was proven to be a $(1-1/e)$-approximation of the IM problem. However, Greedy suffers from two sources of computational deficiency: 1) the need to evaluate many candidate nodes before selecting a new seed in each round, and 2) the calculation of the influence spread of any seed set relies on Monte-Carlo simulations.

In an effort to improve Greedy, Leskovec {\it et al.}~\cite{leskovec07} recognized that not all remaining nodes need to be evaluated in each round and
proposed the ``Cost-Effective Lazy Forward'' (CELF) scheme. Experimental results demonstrate that CELF optimization could achieve as much as 700-time speed-up in selecting seeds. However, even with CELF mechanism, the number of candidate seeds is still large. 
Recently, Goyal {\it et al.} proposed CELF++~\cite{GoyalCELF} that has been shown to run from 35\% to 55\% faster than CELF. However, the improvement comes at the cost of higher space complexity to maintain a larger data structure to store the look-ahead marginal gains of each node.

Chen {\it et al.} devises several heuristic algorithms for influence spread computation~\cite{degreediscountic,pmia,ldag}. In Degree Discount~\cite{degreediscountic}, the expected number of additional vertices influenced by adding a node $v$ in the seed set is estimated based on $v$'s one hop neighborhoods. It also assumes that the influence probability is identical on all edges. In~\cite{pmia} and~\cite{ldag}, two approximation algorithms, PMIA and LDAG are proposed to compute the maximum influence set under IC and LT models, respectively.
In LDAG, it has been proven that under the LT model, computing influence spread in a DAG has linear time complexity, and a heuristic on local DAG construction is provided to further reduce the compute time. We have proven in Section~\ref{sec:spreadondag} that computing influence spread in a DAG under the IC model remains \#P-hard. The marked difference between the two results arises from the fact that in the LT model, the activation of incoming edges is coupled so that in each instance, only one neighbor can influence the node of interest in an equivalent random graph model.

Another line of work explores diffusion models beyond LT and IC. Even-Dar {\it et al.}~\cite{Even-Dar} argue that the most natural model to represent diffusion of opinions in a social network is the probabilistic voter model where in each round, each person changes his opinion by choosing one of his neighbors at random and adopting the neighbor's opinion. Interestingly, they show that the straightforward greedy solution, which picks the nodes in the network with the highest degree, is optimal. Sylvester~\cite{Sylvester09maximizingdiffusion} studies the spread maximization problem on dynamic networks and examines the use of dynamic measures with Greedy algorithm on both LT and IC models. Chen {\it et al.}~\cite{chen09} consider a new model that incorporates negativity bias and design an algorithm to compute influence on tree structures.

Inapproximability results of problems related to IM have also been investigated in literature. MINSEED is the problem of finding the minimized seed set size to activate all or a portion of vertices. Chen~\cite{ningchen} proves that under LT model with a general threshold, MINSEED can not be approximated within a ratio of $O(2^{\log^{1-\varepsilon}n})$, for any fixed $\varepsilon > 0$, unless $NP \subseteq DTIME(n^{polylog(n)})$. In the case when the threshold equals two, the author proves that it is as hard as the case with a general threshold, even for constant degree graphs. Ackerman {\it et al.}~\cite{ackerman} cast MINSEED and IM as maximization problems making them amenable to optimization techniques. However, since the number of variables and constraints grow in $O(n^2)$ and $O(n^3)$ respectively -- $n$ being the number of vertices in the graph -- this approach is only tractable in small-size problems. MINTIME is the problem of finding a target size $k$ such that all or a portion of vertices are activated in the minimum possible time (in terms of spread time or hop count). With a given coverage threshold $\eta$, Goyal {\it et al.}~\cite{goyal} prove that under both IC and LT model, the greedy algorithm can produce the result covering $\eta - \varepsilon$ vertices ($\varepsilon > 0$) in min time, with seed size $|S| \leq k(1+\ln(\eta/\varepsilon))$.  Ni {\it et al.}~\cite{yaodongni} investigate the MINTIME problem by proposing a new spread model and proving various timing bounds on the proposed model.

Literatures on epidemiology are also related to the IM problem that identifies nodes that can initiate viral propagation to most part of the network. Under the proposed model, the authors of~\cite{tisec07} proved that the epidemic threshold for a network is exactly the inverse of the largest eigenvalue of its adjacency matrix. In a follow-up work~\cite{icdm10}, the authors used the previously defined epidemic threshold to quantify the vulnerability of a given network and devised a fast algorithm to choose the best $k$ nodes to be immunized (removed) so as to minimize network vulnerability.~\cite{pkdd10} considered the immunization problem on dynamic networks. The key differences between work on viral immunization literatures and IM lie in the spreading model adopted (e.g.: SIS [susceptible-infected-susceptible] or SIR [susceptible-infected-recovered] vs. IC or LT) and whether the dynamics in the evolution of influence are of interest.

Most existing work on the IM problem only considers cardinality constraints. CELF~\cite{leskovec07} is the only applicable approach to the BIM problem. We will later show in our evaluation that the proposed methods outperform CELF in term of running time (several orders of magnitude faster) and performance on dense networks.

This article is an extended version of our conference paper in~\cite{NguyenIMBP}. We modified our main algorithm to solve the BIM problem and prove its approximation factor. We added detailed algorithm description, complexity analysis, and report more comprehensive results regarding algorithm performance on different real datasets. We also conducted new experiment sets on synthetic networks and provide results on the impact of graph structure on different IM algorithms which, to the best of our knowledge, has never been studied before.

\section{The Budgeted Influence Maximization Problem}
\label{sec:bim}
In this section, we consider the BIM problem with the objective to select the seed set that maximizes influence spread given a fixed budget and arbitrary node costs.

%
\subsection{Problem Formulation}
Consider the network a directed graph $\MG = (V, E)$ with $|V| = n$ vertices and $|E| = m$ edges.
For every edge $(u, v) \in E$, $p(u,v)$ denotes the probability of influence being propagated on the edge.
In this paper, we adopt the Independent Cascade (IC) model. Given a seed set $S \subseteq V$, the IC model works
as follows. Let $S_t\subseteq V$ be the set of node (newly) activated at time
$t$, with $S_0 = S$ and $S_t \cap S_{t-1} = \emptyset$. At round $t+1$, every node $u  \in S_t$ tries to activate its
neighbors in $v \in V\backslash\bigcup_{0\le  i\le t}{S_i}$ independently with
probability $p(u,v)$. The influence spread of $S$, denoted by $\sigma(S)$, is
the {\it expected} number of activated nodes given seed set $S$.

Kempe {\it et al.}~\cite{kempe03} proved two important properties of the $\sigma(\cdot)$ function: 1) $\sigma(\cdot)$ is {\it submodular}, namely, $\sigma(S\cup\{v\}) - \sigma(S) \geq \sigma(T\cup\{v\}) - \sigma(T)$ for all $v\in V$ and all subsets $S$ and $T$ with $S \subseteq T \subseteq V$; 2) $\sigma(S)$ is {\it monotone}, i.e. $\sigma(S) \leq \sigma(T)$ for all set $S \leq T$. For any given spread function $\sigma(\cdot)$ that is both submodular and monotone, the problem of finding a set $S$ of size $k$ that maximizes $\sigma(S)$ can be approximated by a simple greedy approach.

\paragraph*{Budgeted Influence Maximization} In BIM, each node $u$ is associated with an arbitrary cost $c(u)$. The goal is to select a seed set $S \subseteq V$ such that the total cost of this set is less than a budget $b$. Denote by $c(S)$ the total cost of a set, i.e., $c(S)=\sum_{u \in S}{c(u)}$. Budgeted IM (BIM) can be formulated as an optimization problem:
\begin{equation}
\begin{array}{ll}
\underset{S \subseteq V}{\max} & \sigma(S)\\
\mbox{s.t.} &  c(S) \leq b
\end{array}
\end{equation}

When $c(u)\equiv 1, \forall u \in S$, BIM degenerates to the original IM problem. Thus, we call IM the unit-cost BIM. Since IM is NP-hard, it is easy to see that BIM is NP-hard as well.
Key notations used in this paper are summarized in Table~\ref{tab:mainnotations}.

\begin{table}[t]
\caption{Notations}
\centering 
\vspace{0.1in}
{\small
\begin{tabular}{c | c}
\hline
$\MG, V, E$ & the directed graph, its set of vertices and edges\\
\hline
$n, m$ & the number of nodes, edges in $\MG$\\
\hline
$k, b$ & the budget in term of node count and cost\\
\hline
$p(u,v)$ & the propagation probability from $u$ to $v$\\
\hline
$p(v)$ & the activation probability of the node $v$\\
\hline
$c(v)$ & the cost of the node $v$\\
\hline
$Par(v)$ & the set of parents of the node $v$\\
\hline
$S$ & the selected seed set\\
\hline
$\theta$ & the influence threshold\\
\hline
$\sigma(S)$ & the influence spread of the set $S$\\
\hline
& the incremental influence spread of\\[-0.8ex]
\raisebox{1.2ex}{$\delta(v)$}
& selecting $v$ as a seed node\\
\hline
& the directed acyclic graph from $\MG$ on which \\[-0.8ex]
\raisebox{1.2ex}{$\MD(S)$}
& influence is spread given the seed set $S$\\
\hline
\end{tabular}
}
\label{tab:mainnotations}
\end{table}

\subsection{The Seed Selection Algorithm}

First, we consider an intuitive greedy strategy that selects at each step a node $u$ that maximizes the spread gained over cost ratio if the cost of $u$ is less than the remaining budget. We hereby refer to this scheme as the {\it Naive Greedy} approach. Let $r$ be the number of iterations executed and $S_r$ be the seed set at step $r$. Note that $|S_r| \leq r$. At step $r + 1$, Naive Greedy calculates the incremental spread-cost ratio.
\beq
\delta(v) = (\sigma(S \cup v) - \sigma(S))/c(v), \forall v \in V\backslash S.
\eeq
The algorithm chooses $u$ if $u = \argmax_{v \in V, c(s_r \cup v) \le b}{\delta(v)}$. The algorithm terminates when no budget remains, or no node can be added to $S$. Naive Greedy is summarized in Algorithm~\ref{algo:naivegreedy}.

\begin{algorithm}[t]
\caption{Naive Greedy}
\label{algo:naivegreedy}
\SetKwInOut{Input}{input}
\SetKwInOut{Output}{output}
\Input{$G = (V,E), b$}
\BlankLine
\nl $S = \emptyset$\\
\nl \Repeat{ $V =  \emptyset$}
    {
\nl     $\delta(v) = (\sigma(S \cup v) - \sigma(S))/c(v), \forall v\in V$\\
\nl     $u = \argmax_{v \in V}{\delta(v)}$\\
\nl     \If{$c(S \cup u) \leq b$}{
\nl         $S = S \cup u$\\
        }
\nl     $V = V\backslash u$\\
    }
\BlankLine
\Output{$S$}
\end{algorithm}

We first observe that Naive Greedy can have unbounded approximation ratio. Consider a network containing $l+1$ nodes $V = \{u, v_1, v_2, \cdots, v_l\}$. Every pair in $v_1, v_2, \cdots, v_l$ is connected by an edge with influence probability one, while $u$ is an isolated node. Let the cost $c(u) = 1 - \varepsilon$, $c(v_i) = l, \forall i = 1, \cdots, l$ and the budget $b = l$. The optimal solution will pick any node $v_i$ and achieve an influence spread of $l$. In contrast, Naive Greedy picks $u$ since it has the maximum influence-cost ratio $1/1-\epsilon > 1$. The resulting influence spread is 1. Thus, the approximation ratio for Naive Greedy is $l$.

Next, we show that Naive Greedy can be modified to achieve a constant approximation ratio. This algorithm is an adaptation of an algorithm first proposed by Khuller {\it et al.}~\cite{Khuller199939}. We assume that there is no node with a cost greater than the budget $b$, as it will never be a feasible solution to BIM. Let $S_1$ be the seed set selected by Naive Greedy, we consider another candidate solution $s_{max}$, which is the node that has the largest influence. We compare the spread of $S_1$ and $s_{max}$, then output the one with higher influence spread. The process is illustrated in Algorithm~\ref{algo:improvedgreedy}.

\begin{algorithm}[t]
\caption{Improved Greedy}
\label{algo:improvedgreedy}
\SetKwInOut{Input}{input}
\SetKwInOut{Output}{output}
\Input{$G = (V,E), b$}
\BlankLine
\nl $S_1 = $ result of Naive Greedy\\
\nl $s_{max} = \argmax_{v\in V}{\sigma(v)}$\\
\nl $S = \argmax{(\sigma(S_1),\sigma(s_{max}))}$
\BlankLine
\Output{$S$}
\end{algorithm}

\vspace{0.1in}
\begin{thm}
\label{thm2}
Algorithm~\ref{algo:improvedgreedy} provides a $(1-1/\sqrt{e})$-approximation for the BIM problem.
\end{thm}
\vspace{0.1in}

By considering the candidate solution with the maximum influence spread, Algorithm 2 guarantees the approximation ratio within a constant factor, while Algorithm 1 is unbounded.
Note that Algorithm~\ref{algo:improvedgreedy} is different from CELF presented by Leskovec {\it et al.} in~\cite{leskovec07}. CELF runs Naive Greedy on the budgeted and the unit-cost (by setting all costs to one) versions of the problem, and selects the set with the maximum influence spread.
While finding the seed set to maximize IM consumes more time than what it takes to select a single node with the largest spread, CELF can only guarantee a looser bound of $\frac{1}{2}(1-1/e)$ ($\sim 0.316$).

\paragraph*{Complexity} Let $T$ be the maximum time needed to calculate the value of $\sigma(S), \forall S \subseteq V$. Algorithm~\ref{algo:naivegreedy} runs in $O(n^2T)$ time where $n$ is the number of nodes (i.e. $n = |V|$). Finding $S_1$ therefore costs $O(n^2T)$. $s_{max}$ can be determined in in $O(nT)$ time. Algorithm~\ref{algo:improvedgreedy} therefore runs in $O(n^2T)$ time. Note in~\cite{Khuller199939} that Greedy with partial enumeration heuristic can achieve an approximation guarantee of $(1-1/e)$. However, the improvement is attained at the expense of much higher computation complexity of $O(n^4d)$~\cite{chekuri04maximum}.

\vspace{0.1in}
Algorithm~\ref{algo:improvedgreedy} calls $\sigma(.)$ as a subroutine. The efficiency of $\sigma(.)$ computation is thus critical to the overall running time of the algorithm. In the following sections, we develop efficient algorithms for approximating the spread function $\sigma(.)$. We first consider the special case when the network is a directed acyclic graph (DAG). Then, we provide two DAG construction algorithms from a general network graph. Finally, some techniques to further optimize the execution of Algorithm~\ref{algo:improvedgreedy} is presented.

\section{Determining Influence Spread on DAG}
\label{sec:spreadondag}
Given a seed set, estimating value of the $\sigma(.)$ from that seed set was proven to be a \#P-complete problem~\cite{kempe03}. We show in this section that under the IC model, calculation of $\sigma(.)$ remains \#P-complete even when the underlying network graph is a DAG. Then we establish the equivalence between computing $\sigma(.)$ on a DAG and the computation of marginal probabilities in a Bayesian network.

\subsection{Hardness of Computing Influence Spread on DAGs}

In~\cite{kempe03}, Kempe {\it et al.} proposed an equivalent process of influence
spread under the IC model, where at the initial stage, an edge $(u,v)$ in $\MG$
is declared to be {\it live} with probability $p(u,v)$ resulting in a subgraph of
$\MG$. A node $u$ is active if and only if there is at least one path from some
node in $S$ to $u$ consisting entirely of {\it live edges}. In general graphs, the
influencer-influencee relationship may differ in one realization to another for
bi-directed edges. In a DAG, on the other hand, such relationship is fixed and
is independent of the outcome of the coin flips at the initial stage (other than the fact that some
of the edges may not be present). Let $x_u, u \in V$ denotes the binary random
variable of the active state of node $u$, namely, $\Prob{x_u = 1} = p(u)$. For
each node $v$ in $S$, $\Prob{x_v = 1} = 1$. If a node $u \not\in S$ does
not have any parent in $\MG$ then $\Prob{x_u = 1} = 0$.  From $\MG$, the
conditional probability $p(x_u|x_{Par(u)})$ is uniquely determined by the edge
probability, where $x_{Par(u)}$ denotes the states of the parents of node $u$.
In other words, influence spread can be modeled as a Bayesian network.  If node
$u$ does not have any parent, $p(x_u|x_{Par(x_u)}) = p(x_u)$.  The joint
distribution is thus given by,
\beq
p(x_1, x_2, \ldots, x_{n}) = \prod_{i=1}^{n}{p(x_i|x_{Par(x_i)})}.
\label{eq:bayesian}
\eeq

Given the outcome of coin flips $C$, $\sigma_C(S) =\sum_{u \in V}{x_u}$.
Therefore,
\beq
\sigma(S) = \E(\sigma_C(S)) = \sum_{u \in V}{\E(x_u)} = \sum_{u\in
V}{p(u)}.
\label{eq:sigma}
\eeq
The second equality is due to the linearity of expectations. To
compute $p(u)$, we can sum \eqref{eq:bayesian} over all possible configurations
for $x_v, v \in V\backslash u$. Clearly, such a naive approach has complexity
that is exponential in the network's treewidth. In fact, the marginalization
problem is known to be \#P-complete on a DAG. However, since computing
influence spread on a DAG can be reduced to a special instance of the
marginalization problem, it remains to be shown if the former problem is
\#P-complete. The main result is summarized in the following theorem\footnote{All proofs are presented in the Appendix}.

\begin{thm}
\label{thm:sharpphard}
Computing the influence spread $\sigma(S)$ on a DAG given a seed set $S$ is \#P-complete.
\end{thm}

\subsection{Estimating $\sigma(\cdot)$ via Belief Propagation}
Belief propagation (BP) is a message passing algorithm for performing inference on
graphical models, such as Bayesian networks and Markov random fields. It
calculates the marginal distribution for each unobserved node, conditional on
any observed nodes~\cite{understandingbp}. For {\it singly-connected} DAGs, where between any
two vertices there is only one simple path, the BP
algorithm in~\cite{pearl} computes the exact solution with $O(n)$ complexity. For
multi-connected DAGs, where  multiple simple paths may exist between two
vertices, belief propagation and many of its variants~\cite{understandingbp,junctiontree,loopy} have been shown to work well in general. Exact solutions
such as junction tree~\cite{junctiontree} may incur the worst case complexity
exponential to the number of vertices  due to the need to enumerate all cliques
in the DAG.


BP algorithms take as input a factor graph or a description of the underlying Bayesian Network. In the context of influence spreading, each node only has two states: active and inactive. BP algorithms calculate the probability of each node in either states. $\sigma(\cdot)$ can then be determined by summing up the probability of nodes being active.

\paragraph*{Computation complexity}
The complexity of $\sigma(\cdot)$ calculation is dominated by the execution of
the BP algorithm. A variety of BP algorithms exist. In this work, we adopt
the Loopy Belief Propagation (LBP) algorithm which was shown to perform well for various
problems~\cite{Frey01veryloopy, turbodecoding}. LBP takes $O(M^d)$ to estimate
the active probability of a node, where $M$ is the number of
possible labels (states) for each variable ($M = 2$), and $d$ is the maximum in-degree.
We denote by $n_0$ the number of vertices in a DAG. Thus, the complexity of
LBP is $O(n_02^d)$.

\subsection{A Single Pass Belief Propagation Heuristic for $\sigma(\cdot)$ Estimation}
Calculating $\sigma(\cdot)$ with LBP produces highly accurate
results, but the computation time remains to be high when the graph is
multi-connected. The main complexity arises from the fact that the activation
of parents of a node may be correlated in a multi-connected graph. Thus, in
computing the activation probability of the node, one needs to account for the
joint distribution of its parent nodes. Next, we propose a single pass belief
propagation (SPBP) algorithm that ignores such correlation in determining
$\sigma(\cdot)$. Note that the heuristic is exact when the graph is
singly-connected.

Let $\MD(\cdot)$ be the input DAG. Consider a node $v \in \MD(\cdot)$. Given the activation probabilities of its parents $Par(v)$, we approximate $p(v)$ as,
\beq
p(v) = 1 - \prod_{u \in Par(v)}(1 - p(u)p(u,v)).
\eeq
The algorithm is summarized in Algorithm~\ref{algo:polynomialp}. It starts with the seed nodes and proceeds with the topological sorting order. The total complexity is $O(n_0d)$. Clearly, SPBP is much faster than LBP.
\begin{algorithm}[t]
\small
\caption{Single-Pass Belief Propagation (SPBP)}
\label{algo:polynomialp}
\SetKwInOut{Input}{input}
\SetKwInOut{Output}{output}
\Input{$\MD(S)$}
\BlankLine
\nl $\sigma(S)$ = 0;\\
\nl \ForEach{$v \in \MD(S)$} {
\nl     \If {$v \in S$}{
\nl         $p(v) = 1$\\
        } \Else {
\nl         $p(v) = 1 - \prod_{u \in Par(v)}(1 - p(u)p(u,v))$
        }
\nl     $\sigma(S) = \sigma(S) + p(v)$
    }
\Output{$\sigma(S)$}
\end{algorithm}

\section{DAG Construction}
\label{sec:dagmodels}
In general, real social networks are not DAGs (with the exception of
advisor-advisee and parent-child relationship, for instance, which exhibit a
natural hierarchy). To apply the BP algorithm in computing influence spread, one
needs to selectively prune edges and reduce the graph to a DAG. Clearly, there
are many ways to do so. The challenge is to find a DAG that approximates well
the original graph in influence spread. In this section, we introduce two DAG
construction algorithms, both retaining important edges where influences are likely
to travel.

\subsection{Localizing Influence Spread Region}
\label{sec:daglocalize}
One important observation in \cite{pmia} is that the influence of a seed node
diminishes quickly along a path away from the seed node. In other words, the
``perimeter" of influence or the {\it influence region} of a seed node is in fact
very small. One way to characterize the {\it influence region} of a node $v$ is
through the union of maximum influence paths defined next.

\begin{defn} (Path Propagation Probability) \\
For a given path $P(u,v) = \{u_1, u_2, \ldots, u_l\}$ of length $l$ from a vertex $u$ to $v$, with $u_1 = u, u_l = v$ and $u_2,\ldots,u_{l-1}$ are intermediate vertices, define the propagation probability of the path, $p(P)$, as:
\begin{equation}
p(P(u,v)) = \prod_{i=1}^{l-1}p(u_1, u_{i+1}).
\end{equation}
$p(P(u,v))$ can be thought as the probability that $u$ will influence $v$ if $u$ is selected as a seed node.
Obviously, the longer the path length $l$, the smaller the chance that $u$ can spread its influence to $v$.
\end{defn}

\begin{defn} (Maximum Influence Path) \\
Denote by $\MP(\MG,u,v)$ the set of all paths from $u$ to $v$ in $\MG$. The maximum influence path $MIP(\MG,u,v)$ from $u$ to $v$ is defined as:
\begin{equation}
MIP(\MG,u,v) = \arg\max_P\{p(P)|P\in \MP(\MG, u, v)\}.
\end{equation}
Ties are broken in a predetermined and consistent way such that $MIP(\MG,u,v)$ is always unique, and any sub-path in $MIP(\MG,u,v)$ from $x$ to $y$ is also the $MIP(\MG,x,y)$. In order to localize the influence region of nodes and reduce the complexity, we only consider influence spread on maximum influence paths.
\end{defn}

\begin{defn} (Maximum Influence Out-Arborescence) \\
For a graph $\MG$, an influence threshold $\theta$, the maximum influence out-arborescence of a node $u \in V, MIOA(\MG, u, \theta)$, is defined as:
\begin{equation}
MIOA(\MG,u,\theta) = \bigcup_{v\in V, p(MIP(\MG,u,v)) \geq \theta}MIP(\MG,u,v).
\end{equation}
\end{defn}

$MIOA(\MG,u,\theta)$ is defined as the union of $MIP$'s from $u$ to all other
nodes in $V$. $MIP$'s with propagation probabilities less than a threshold
$\theta$ are not included to reduce the size of $MIOA$. One can think of
$MIOA(\MG,u,\theta)$ as a {\it local region} where $u$ can spread its influence
to. $MIOA(\MG,u,\theta)$ can be computed by first finding the Dijkstra tree
rooted at $u$ with edge weight $-\log(p(u,v))$ for edge $(u,v)$, and then
removing the paths whose cumulative weights are too high. By tuning the
parameter $\theta$, influence regions of different sizes can be obtained.
For a single node, its MIOA is clearly a tree. For multiple seed nodes,
we build upon the idea of local influence region and propose two algorithms.

\subsection{Building DAGs from a Seed Set}
\paragraph*{DAG 1}
We observe that any DAG has at least one topological ordering. Conversely,
given a topological ordering, if only edges going from a node of low rank to
one with high rank are allowed, the resulting graph is a DAG.

To obtain the topological ordering given seed set $S$, we first introduce a
(virtual) super root node $R$ that is connected to all seed nodes with edge
probability 1. Let
$\MG_R = (V_{\MG_R}, E_{\MG_R})$ where $V_{\MG_R} = V \cup \{R\}$ and $E_{\MG_R}
= E \cup \{(R,S_k)|\forall S_k \in S\}$. We build $MIOA(\MG_R,R,\theta)$ by
calculating a Dijkstra tree from $R$. After removing $R$ and its edges from
$MIOA(\MG_R,R,\theta)$, we obtain a singly connected DAG $\MD_1 = (V_{\MD_1},
E_{\MD_1})$ on which BP algorithms can be directly applied and used to estimate the
influence spread from $S$. However, $\MD_1(\cdot)$ is very sparse (with $n-k$ edges)
since many edges are removed.

We then augment $\MD_1(\cdot)$ with additional edges. Note that $MIOA(\MG_R,R,\theta)$
provides a topology ordering. More specifically, let the rank of node $v$ be
the sum weight of the shortest path from $R$, namely,

\begin{equation}
\label{eq:weightlist}
r(v) = \min(-\log(p(P(s,v)))), \forall s \in S.
\end{equation}

Rank grows as the node is further away from $R$.
We include in $\MD_1(\cdot)$ all edges in $\MG$ whose end points are in
$\MD_1(\cdot)$ and go from a node with lower rank to one with higher rank.
Clearly, the resulting graph is a DAG.  The  DAG constructing
procedure is illustrated in Figure~\ref{fig:dag1} and summarized in
Algorithm~\ref{algo:dag1}.

\begin{figure}[h]
\begin{center}
\includegraphics[width=3.6in]{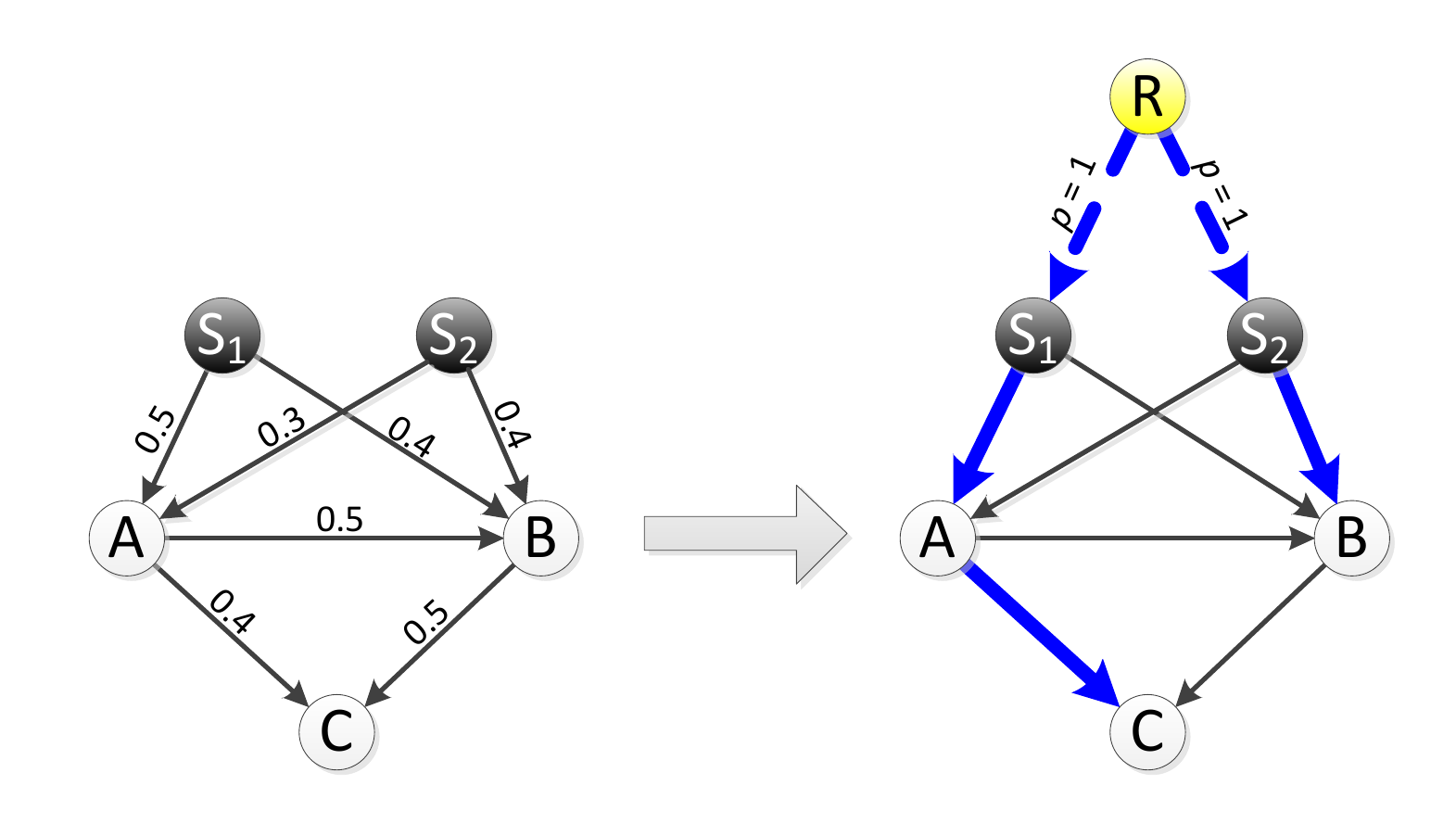}
\\
\centering
{\normalsize
\begin{tabular}{ c | c c c c c}
Node & $S_1$ & $S_2$ & $A$ & $B$ & $C$ \\
\hline
$r($Node$)$ & 0 & 0 & 0.301 & 0.398 & 0.699
\end{tabular}
}
\caption{DAG due to Algorithm~\ref{algo:dag1}. $S_1$ and $S_2$ are seed nodes.
Edges in $MIOA(\MG_R,R,\theta)$ are in bold. $(S_1, B)$, $(S_2, A)$, $(A, B)$,
and $(B, C)$ are added into $\MD_1(S)$ to improve inference accuracy.
$\theta = 0.0001$.}
\label{fig:dag1}
\end{center}
\end{figure}


\begin{algorithm}[t]
\small
\caption{Calculate $\MD_1(S)$ from a seed set $S$}
\label{algo:dag1}
\SetKwInOut{Input}{input}
\SetKwInOut{Output}{output}
\Input{$\MG, S, \theta$}
\BlankLine
\nl Build $\MG_R = (V_{\MG_R}, E_{\MG_R})$\\
\nl $\MD_1(S) = MIOA(\MG_R,R,\theta) \backslash R$\\
\nl Calculate $r(v), \forall v \in V_{\MD_1}$ (Eq.~\eqref{eq:weightlist})\\
\nl \ForEach{$(u,v) \in V_{\MG_R}$} {
\nl     \If {$r(u) < r(v)$ and $(u,v) \in E$}{
\nl         $\MD_1(S) = \MD_1(S) \cup (u,v)$\\
        }
    }
\Output{$\MD_1(S)$}
\end{algorithm}


\paragraph*{DAG 2}
In the second algorithm, we first compute the $MIOA$ from each seed node and
take the union of $MIOA(\MG,s,\theta),\forall s \in S$. Denote the resulting
graph $\MD_2(S) = (V_{\MD_2}, E_{\MD_2})$. Note that $\MD_2(S)$ is not necessary
a DAG as there could be circles.
 To break the cycles, certain edges need to be removed. We adopt a similar approach
as in Algorithm~\ref{algo:dag1}. A node $v$ is associated with a rank $r(v)$ as
in (\ref{eq:weightlist}). Only edges that connect a lower ranked node to higher ranked node
are retained. Clearly, the resulting graph is a DAG. The approach is
summarized in Algorithm~\ref{algo:dag2}.

\begin{algorithm}[t]
\small
\caption{Calculate $\MD_2(S)$ from a seed set $S$}
\label{algo:dag2}
\SetKwInOut{Input}{input}
\SetKwInOut{Output}{output}
\Input{$\MG, S, MIOA(\MG,v,\theta), \forall v \in V$}
\BlankLine
\nl $\MD_2(S) = \bigcup_{\forall s \in S}MIOA(\MG,s,\theta)$\\
\nl Calculate $r(v), \forall v \in V_{\MD_2}$ (Eq.~\eqref{eq:weightlist})\\
\nl \ForEach{$(u,v) \in \MD_2(S)$} {
\nl     \If {$r(u) \geq r(v)$}{
\nl         $\MD_2(S) = \MD_2(S) \backslash (u,v)$\\
        }
    }
\Output{$\MD_2(S)$}
\end{algorithm}


\begin{figure}[h]
\begin{center}
\includegraphics[width=3.6in]{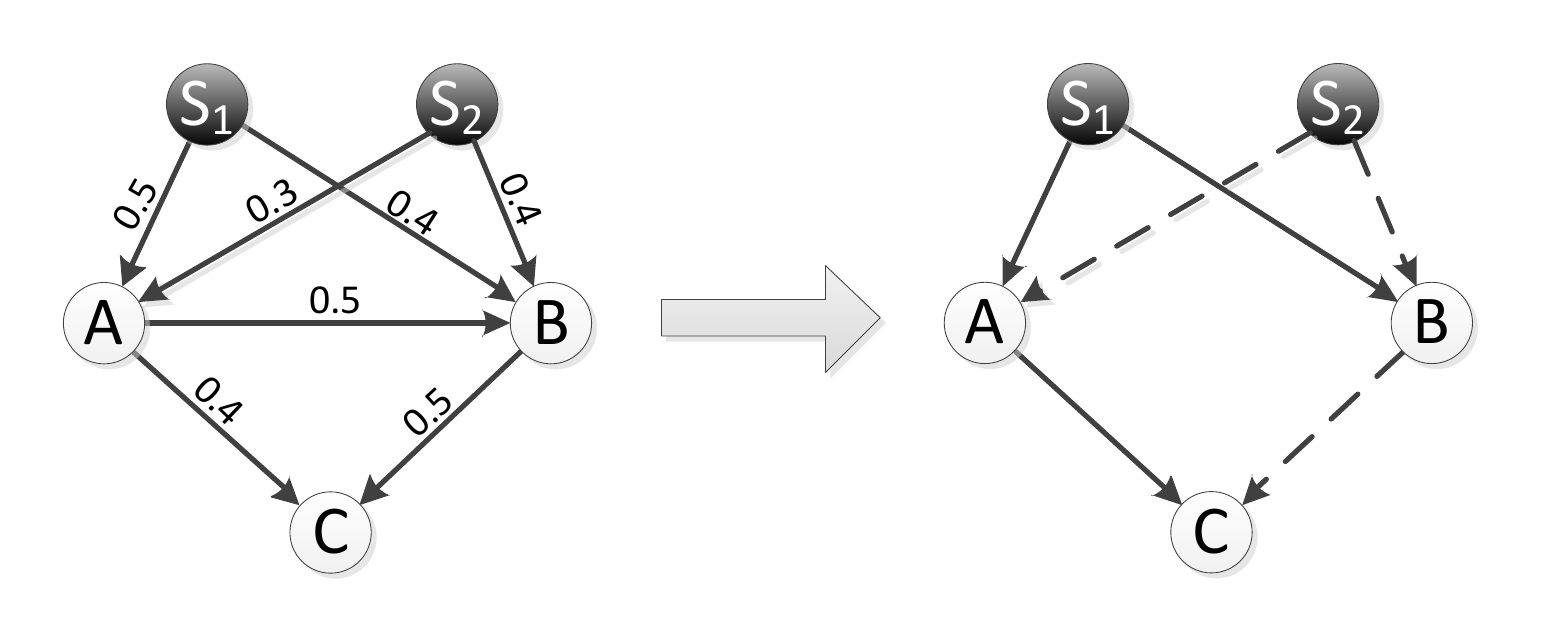}
\\
\centering
{\normalsize
\begin{tabular}{ c | c c c c c}
Node & $S_1$ & $S_2$ & $A$ & $B$ & $C$ \\
\hline
$r($Node$)$ & 0 & 0 & 0.301 & 0.398 & 0.699
\end{tabular}
}
\caption{DAG due to Algorithm~\ref{algo:dag2}. $S_1$ and $S_2$ are seed
nodes. $\MD_2(S)$ is the union of $MIOA(\MG,S_1,\theta)$ (solid edges) and
$MIOA(\MG,S_2,\theta)$ (dashed edges). $\theta = 0.0001$.}
\label{fig:dag2}
\end{center}
\end{figure}

The next proposition provides the relationship between DAGs constructed by
Algorithm~\ref{algo:dag1} and~\ref{algo:dag2}.

\begin{prop}
\label{thm:dag1vsdag2}
Given a fixed influence threshold $\theta$, let $\MD_1(\cdot) =
(V_{\MD_1},E_{\MD_1})$ and $\MD_2(\cdot) = (V_{\MD_2},E_{\MD_2})$ be the DAGs
constructed by Algorithm~\ref{algo:dag1} and Algorithm~\ref{algo:dag2}. Then,
$V_{\MD_1} = V_{\MD_2}$ and $E_{\MD_2} \subseteq E_{\MD_1}$.
\end{prop}

\paragraph*{Computation complexity}
Building the Dijkstra tree from a source node takes $O(n_0\log{n_0})$, where $n_0$ is the maximum number of vertices in the resulting DAG. Calculating the node rank $r(\cdot)$ takes $O(n_0)$, the union operation in DAG 2 takes $O(n_0 -1)$, and the edge augmenting and pruning in DAG 1 and DAG 2 takes $O(m_0)$ and $O(\min(m_0, k(n_0 - 1)))$, respectively, where $m_0$ is the maximum number of edges in a DAG and $k$ is the seed set cardinality.

Thus, the running time of DAG 1 and DAG 2 are $O(n_0\log{n_0})$ and $O(n_0)$, respectively. Note that DAG 2 calculation requires the availability of $MIOA(\MG,v,\theta), \forall v \in V$ first, which can be built at the initialization stage at the cost of $O(nn_0\log{n_0})$. Assuming that $k$ is small and $\theta$ is properly selected, we have $n0 \ll n$.

\section{Optimization of Seed Selection}
\label{sec:proposedalgorithm}



In each round of Naive Greedy, a seed node with the maximum
incremental spread-cost ratio is selected, namely, $v = \max_{v \in
V\backslash S}\delta(v)$. Recall that $\delta(v) =
(\sigma(S\cup v) - \sigma(S))/c(v)$ is the spread increment ratio of $v$ under $S$.
Initially, when $S = \emptyset$, $\delta(v) = \sigma(v)/c(v)$.
Evaluating $\delta(v)$ at each iteration for all $v \in V$ dominates the overall computation complexity.

To accelerate the execution of Naive Greedy, one can try to improve on
two aspects, namely, 1) limiting the candidate set of nodes to pick from for the next
seed, and 2) reducing the complexity of computing the spread increments. CELF
algorithm~\cite{leskovec07} eliminates many nodes from being evaluated. We focus on the
second aspect. The proposed mechanism can be used in conjunction with the idea from CELF.

Recall in Section~\ref{sec:daglocalize}, we use $MIOA$ to localize the influence region of a
node. Consider for now that influence from a node can only reach nodes in its
$MIOA$. Then, we make the following claim.

\vspace{0.05in}
\begin{prop}
\label{prop:local}
Given the current seed set $S$, adding $u$ to $S$ will not change the spread
increment of $v$, namely, $\delta_S(v) = \delta_{S\cup u}(v)$ if
$MIOA(\MG,u,\theta)$ and $MIOA(\MG,v,\theta)$ have no common vertex.
\end{prop}
\vspace{0.05in}

As a result of Proposition~\ref{prop:local}, each time we select a new seed,
only the influence increments of nodes that have overlapping influence regions
with the newly selected seed need to be re-evaluated. Formally, we define the
set of Peer Seeds (PS) of a vertex $v\in V$ as follow:
\begin{equation}
PS(\MG, v,\theta) = \left\{ s \in V | MIOA(\MG, s,\theta)\cap MIOA(\MG, v, \theta) \neq \emptyset \right\}.
\end{equation}

$PS(\MG, v,\theta)$ can be computed efficiently just once at the
beginning when all $MIOA(\MG,v,\theta)$'s are available.

Combining the ideas of 1) limiting the region to be re-evaluated using $PS$, 2) limiting the set of nodes to pick from (adopted from CELF), and 3) picking nodes w.r.t its cost and the remaining budget (Algorithm~\ref{algo:improvedgreedy}), we have the complete procedure to determine the optimal seed set in Algorithm~\ref{algo:newgreedy}.
Figure~\ref{fig:blockdiagram} gives the block diagram of the proposed algorithm.

\begin{figure}[t]
\begin{center}
\includegraphics[width=3.6in]{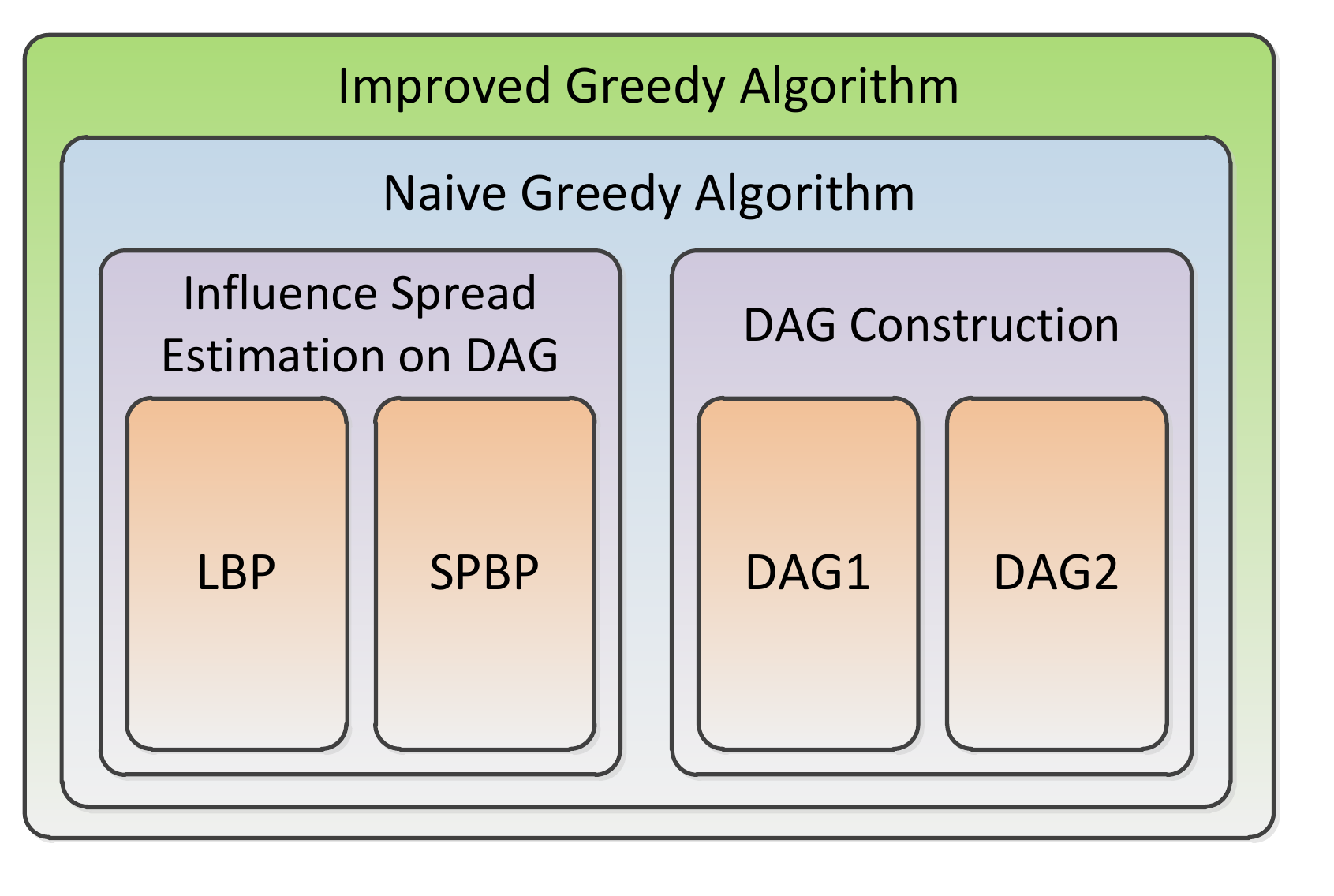}
\caption{The building blocks of our proposed algorithm. Details are presented in the previous sections.}
\label{fig:blockdiagram}
\end{center}
\end{figure}


\begin{algorithm}[h]
\small
\caption{The Proposed Algorithm}
\label{algo:newgreedy}
\SetKwInOut{Input}{input}
\SetKwInOut{Output}{output}
\Input{network graph $\MG(V,E)$ and budget $b$}
\BlankLine
    // {\it initialization}\\
\nl $S = S_1 = s_{max} = \emptyset, \sigma_0 = 0, \theta =$ influence threshold\\
\nl \ForEach{$v \in V$} {
\nl     build $MIOA(\MG, v, \theta)$\\
\nl     $\MD(v) = MIOA(\MG, v, \theta)$\\
\nl     calculate $\sigma(v)$ (LBP or Algorithm~\ref{algo:polynomialp})\\
\nl     $\delta(v) = \sigma(v) / c(v)$\\
\nl     $\delta_{old}(v) = 0$\\
    }
\nl build $PS(\MG, v, \theta), \forall v \in V$\\
    \BlankLine
    // {\it select $s_{max}$}\\
\nl $s_{max} = \argmax_{v\in V}{\sigma(v)}$\\
    \BlankLine
    // {\it select $S_1$}\\
\nl \While{true} {
        // {\it select a new seed}\\
\nl     $u = \arg\max_{v\in V \backslash S}(\delta(v))$\\
\nl     \If {$c(S_1 \cup u) \leq B$} {
\nl         $S_1 = S_1 \cup \{u\}$\\
\nl         $\sigma_0 = \sigma(S)$\\
\nl         $\delta_{old}(v) = \delta(v), \forall v \in V \backslash S_1$\\
\nl         $b = b - c(u)$ \\
            // {\it update incremental influence spread}\\
\nl         $\delta_{max} = 0$\\
\nl         \ForEach{$v \in PS(\MG, u,\theta) \backslash S_1$} {
\nl             \If {$\delta_{old}(v) > \delta_{max}$} {
\nl                 build $\MD(S_1 \cup \{v\})$ (Algorithm~\ref{algo:dag1} and~\ref{algo:dag2})\\
\nl                 calculate $\sigma(S_1 \cup \{v\})$ (LBP or Algorithm~\ref{algo:polynomialp})\\
\nl                 $\delta(v) = (\sigma(S_1 \cup \{v\}) - \sigma_0) / c(v)$\\
\nl                 \If {$\delta(v) > \delta_{max}$} {
\nl                     $\delta_{max} = \delta(v)$\\
                    }
                }
            }
        }
\nl     $V = V \backslash u$\\
\nl     \If {$V = \emptyset$ or $b = 0$} {break}
    }
\BlankLine
\nl $S = \argmax{(\sigma(S_1),\sigma(s_{max}))}$\\
\BlankLine
\Output{selected seed set $S$}
\end{algorithm}
\normalsize

The seed selection algorithm proceed as follow:
In the initialization phase (lines 1 -- 8), $MIOA$'s and $PS$'es are constructed.
The second candidate solution $s_{max}$ can be determined in $O(n)$ time (line 9). $S_1$ is computed by executing the loop in lines 10 -- 26. Each node in $V$ is ranked by its incremental spread-cost ratio and can be added to $S_1$ just once. The node with the highest ratio is included in $S_1$ if it does not violate the budget $b$ (line 12), and the corresponding nodes will be re-evaluated (lines 18 -- 24). The procedure terminates once all nodes were considered, or no more budget remains (line 26). Finally, the algorithm compares the spread of $S_1$, $s_{max}$ and returns the solution with the larger spread.

\paragraph*{Computation complexity}
Recall that we denoted by $n_0$ the largest number of vertices, and by $d$ the largest in-degree of a node in a DAG. For each node $v \in V$ in the initialization phase, building $MIOA(\MG, v, \theta)$ takes $O(n_0\log{n_0})$, and estimating $\sigma(v)$ takes $O(n_0d)$ using SPBP and $O(n_02^d)$ using LBP, respectively. Thus, depending on the algorithm used, the running time of initialization is $O(nn_0(\log{n_0} + d))$ or $O(nn_0(\log{n_0} + 2^d))$.

Let $k$ be the number of seeds selected in the main loop (lines 10 -- 26) and $v_0$ be the cardinality of the largest set of peer seeds, namely, $v_0 = \max_{\forall v \in V}\{|PS(\MG,v,\theta)|\} = O(n_0)$. Therefore, nodal influence spread is updated $O(kn_0)$ times. Note that this is much less than the number of updates required by Algorithm 1 ($O(n^2)$) as we do not naively re-evaluate every node.
Each time when the influence spread is updated, we need to rebuild the DAG (line 20 -- takes $O(n_0\log{n_0})$ with DAG 1 or $O(n_0)$ with DAG 2) and calculate the influence spread (line 21 -- takes $O(n_02^d)$ with LBP or $O(n_0d)$ with SPBP).
The total computation complexity for different combinations of algorithms is summarized as follows:

\vspace{0.1in}
\hspace{-0.2in}
\begin{tabular}{c|c|c}
 & DAG 1 & DAG 2 \\
\hline
LBP & \footnotesize$n_0(n + kn_0)(\log{n_0} + 2^d)$ & \footnotesize$n_0(2^d(kn_0 + n) + n\log{n_0})$ \\
\hline
SPBP & \footnotesize$n_0(n + kn_0)(\log{n_0} + d)$ & \footnotesize$n_0(n(\log{n_0} + d) + kn_0(1+d))$
\end{tabular}
\vspace{0.1in}

Clearly, combining DAG 1 and LBP incurs the highest complexity while the combination of DAG 2 and SPBP is the fastest.
From the analysis, it is easy to see that the computation complexity depends on $n_0$ and $d$. The proposed approach is more efficient with smaller $n_0$ and $d$; that is, when the graph is sparse and the edge propagation probabilities are small, both are likely true in social networks.

\section{Evaluation}
\label{sec:eval}
In this section, we evaluate the performance of the proposed framework. First, an illustrative example is provided to highlight the difference in the two DAG construction models, and spread computation methods. Next, we present the implementation details and experimental setup. Finally, we present the results on 1) performance on real-world social networks and 2) impact of network structures using synthetic graphs.

\subsection{An Illustrative Example}

\begin{figure*}[t]
\hspace{-0.2in}
\begin{tabular}{ c c c }
\includegraphics[width=2.3in]{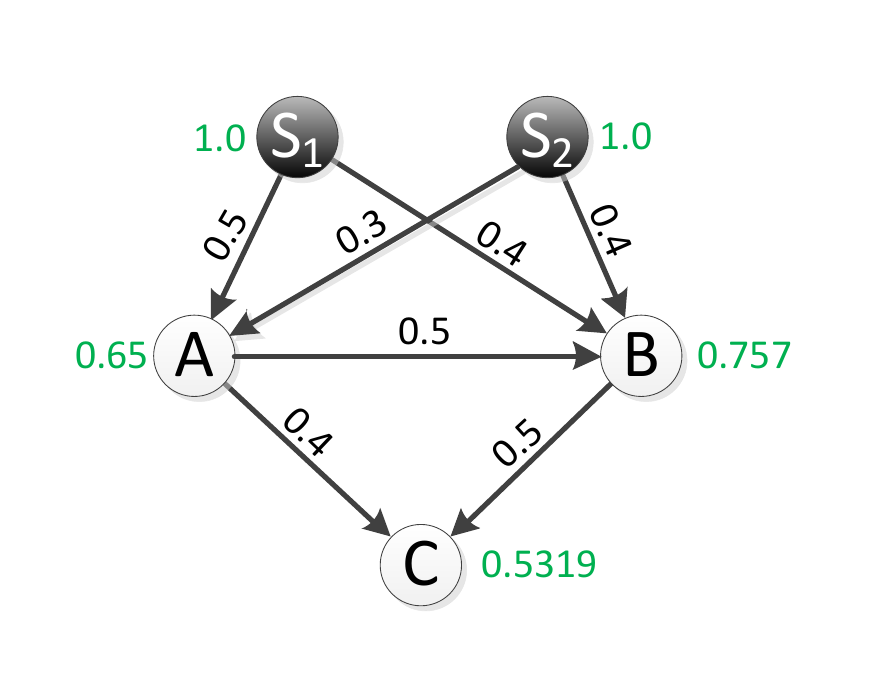} & \includegraphics[width=2.3in]{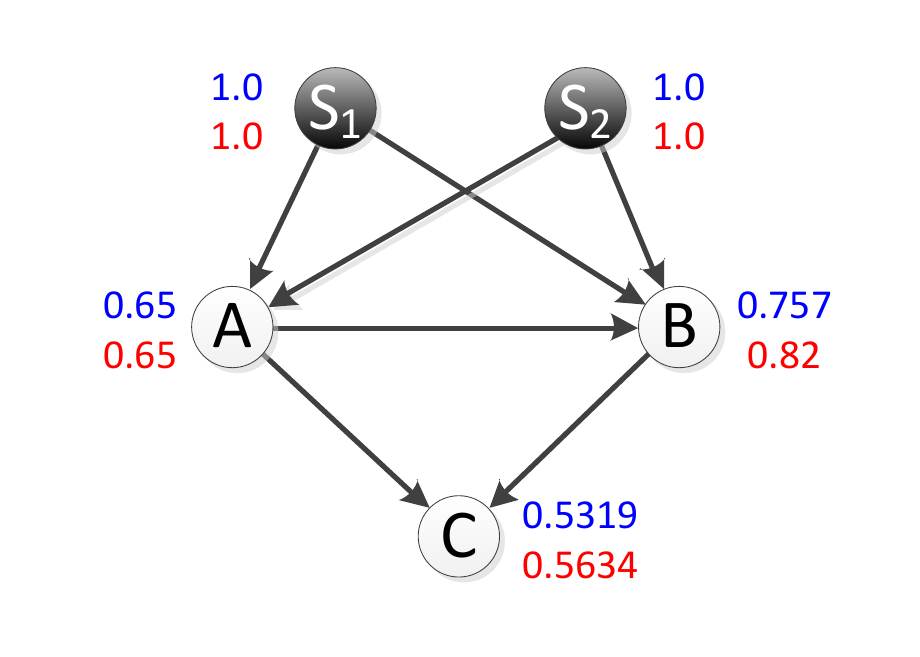} & \includegraphics[width=2.3in]{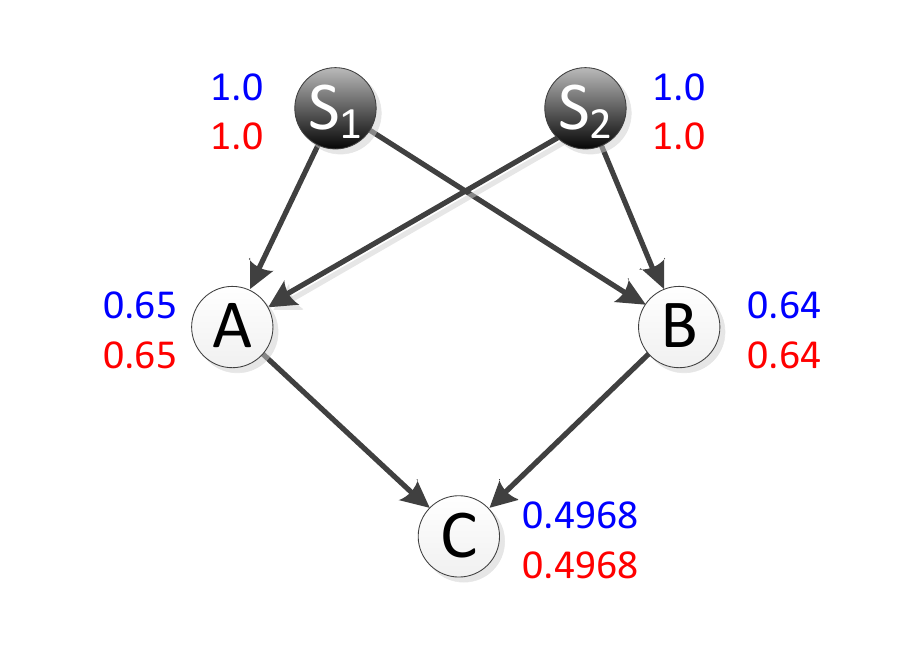}\\
(a) Real active probabilities & (b) Inference on DAG 1 & (c) Inference on DAG 2
\end{tabular}
\caption{Inference result on 2 DAG models. The real active probabilities are in green, LBP results are on top, in blue, and SPBP results are below, in red.}
\label{fig:daginference}
\end{figure*}


Here, we consider a small scale network as shown in Figure~\ref{fig:daginference}(a). Figure~\ref{fig:daginference}(b) and (c) show the DAG constructed by the two models, and the activation probabilities by the two methods. DAG 1 retains all the edges in the network (since the original graph is in fact a DAG), while DAG 2 has fewer edges. When LBP is used to compute the influence spread (the numbers on top next to each node), DAG 1 yields higher activation probability compared to DAG 2 for node B and node C since A has a large influence to B (0.5), which is  not considered in DAG 2. In both DAGs, ignoring the possible correlation among parent nodes in SPBP, the activation probabilities tend to be bigger. Interestingly, though DAG 2 is a multi-connected graph, the activation probabilities computed by both methods are identical. Upon a close examination, we find that even though the graph is multi-connected, the activations of A and B are in fact independent since both are direct descendents of seed nodes with activation probability one.

%
%


\subsection{Experiment Setup}
\paragraph*{The algorithms and implementation}
In addition to the two DAG models and two methods to compute influence spread (a total of 4 combinations DAG1--LBP, DAG1--SPBP, DAG2--LBP, and DAG2--SPBP), we make comparison with the following algorithms:

\begin{itemize}
\item {{\it PMIA($\theta$)}}~\cite{pmia}: a very fast heuristic algorithm that builds a tree-like structure model on which influence is spread. $\theta$ is the influence threshold. We will set $\theta = 1/160$ in all experiments as it was reported to yield the best performance. The PMIA implementation provided by the authors is optimized for IM, and thus its performance for BIM is excluded. 
\item {{\it Greedy/CELF}}: The greedy approach from~\cite{kempe03} with CELF optimization in~\cite{leskovec07}. The number of simulation rounds for each $\sigma(\cdot)$ estimation is 10,000.
\item {{\it Weighted Degree}}: The simple heuristic that selects $k$ seeds that have maximum total out-connection weight. Weighted Degree has been reported to be working very well in practice.
\end{itemize}

We do not compare with other heuristics such as SP1M, SPM~\cite{SP1M}, PageRank~\cite{pagerank}, Random, DegreeDiscountIC~\cite{degreediscountic} or Betweenness centrality~\cite{betweenesscentrality} since they have been reported in previous studies~\cite{pmia,kempe03} to be either unscalable or have poorer performance.

We have implemented the proposed algorithms in C++. All experiments are conducted on a workstation running Ubuntu 11.04 with an Intel Core i5 CPU and 2GB memory. In order to implement LBP algorithm, we use libDAI~\cite{libdai} and Boost~\cite{boost} libraries. We find out through the implementation that running LBP on networks with high in-degree nodes is very costly. Therefore when running LBP, we prune incoming edges on high in-degree nodes such that only ten edges with the highest propagation probabilities are retained. The implementation of PMIA is obtained from its authors. Note that with code optimization, the running time of our algorithms can be further reduced.


\begin{table*}[t]
\caption{Network datasets}
\vspace{0.1in}
\begin{tabular}{|c||c|c|c|c|c|c|c|}
\hline
Name & Type & Nodes & Edges & Density & Max Degree & Mean Degree & Description\\
\hline \hline
 &  &  &  &  &  &  & Email communication within\\[-1ex]
\raisebox{1.2ex}{{\it Email}} & \raisebox{1.2ex}{Email exchange network} & \raisebox{1.2ex}{447} & \raisebox{1.2ex}{5,731} & \raisebox{1.2ex}{0.04} & \raisebox{1.2ex}{195} & \raisebox{1.2ex}{25.64} & a research lab during a year\\
\hline
 &  &  &  &  &  &  & Gnutella peer to peer\\[-1ex]
\raisebox{1.2ex}{{\it p2p-Gnutella}} & \raisebox{1.2ex}{P2P network} & \raisebox{1.2ex}{6,301} & \raisebox{1.2ex}{20,777} & \raisebox{1.2ex}{1e-03} & \raisebox{1.2ex}{97} & \raisebox{1.2ex}{6.59} & network from August 8 2002\\
\hline
 &  &  &  &  &  &  & Slashdot social network\\[-1ex]
\raisebox{1.2ex}{{\it soc-Slashdot}} & \raisebox{1.2ex}{Social network} & \raisebox{1.2ex}{82,168} & \raisebox{1.2ex}{948,464} & \raisebox{1.2ex}{1.6e-03} & \raisebox{1.2ex}{5,064} & \raisebox{1.2ex}{23.09} & from February 2009\\
\hline
 &  &  &  &  &  &  & Amazon product co-purchasing\\[-1ex]
\raisebox{1.2ex}{{\it Amazon}} & \raisebox{1.2ex}{Product co-purchasing network} & \raisebox{1.2ex}{262,111} & \raisebox{1.2ex}{1,234,877} & \raisebox{1.2ex}{2.6e-05} & \raisebox{1.2ex}{425} & \raisebox{1.2ex}{9.42} & network from March 2 2003\\
\hline
\end{tabular}
\label{tab:dataset}
\end{table*}

\paragraph*{Datasets} We use four real-world network datasets from~\cite{snap} and~\cite{emaildataset} to compare performance of different algorithms. The four datasets were selected so as they are representative of the structural features of large-scale social networks, and are of different scales -- from several thousands to millions of edges. The first one is an email exchange network in a research lab, denoted by {\it Email}. Each researcher is a vertex and an email from a researcher $u$ to $v$ constitutes an edge. The second network, denoted by {\it p2p-Gnutella} is a snapshot of the Gnutella peer-to-peer file sharing network from August 2002. Nodes represent hosts in the Gnutella network and edges represent connections between the Gnutella hosts.
The third network comes from Slashdot.org, a technology-related news website, denoted by {\it soc-Slashdot}. In 2002, Slashdot introduced the Slashdot Zoo feature that allows users to tag each other as friends or foes. The network contains friend/foe links between Slashdot users obtained in February 2009. Finally, {\it Amazon} dataset is the product co-purchasing network collected by crawling Amazon website on March 2, 2003. Details of the datasets are summarized in Table~\ref{tab:dataset}.

In addition to real social networks, we modified DIGG~\cite{digg} source code and generated scale-free networks with different network densities and node out-degree distributions.
The purpose which allows us to study the impact of graph structures and network property on the algorithm performance.

\paragraph*{Probability generation models}
Two models that have been used in previous
work~\cite{kempe03, pmia} are: 1) the Weighted Cascade (WC) model where $p(u,v) =
1/d(v)$ where $d(v)$ is the in-degree of $v$ and 2) the Trivalency (TV) model where
$p(u,v)$ is assigned a small value for any $(u,v) \in E$. We argue that both
models are not truthful reflections of the probability model in practice. The WC
model assign a very high probability for a connections to nodes with small
number of incoming connections while the TV model assigns a similar
probability to all edges. In the evaluation, we consider two additional models: 1) Random (RA) where
$p(u,v)$ is randomly selected in the range [0.001, 0.2]. RA is useful when no prior information regarding the influence is available; and 2) Power Law (PL) where $p(u,v)$ follows
the power law distribution with the density function $p(x) = \alpha/x^{\beta}$, with $x$ be the propagation probability between two random edges $p(u,v)$. Parameters $\alpha = 0.05$ and $\beta = 0.9$ are selected so that $p(u,v)$ has the mean value 0.1 in the range [0.001, 0.2].

\subsection{Real Social Networks}

\paragraph*{Unit-cost version of BIM}

\begin{figure*}[t]
\hspace{-0.3in}
\begin{tabular}{cccc}
\includegraphics[width=1.75in]{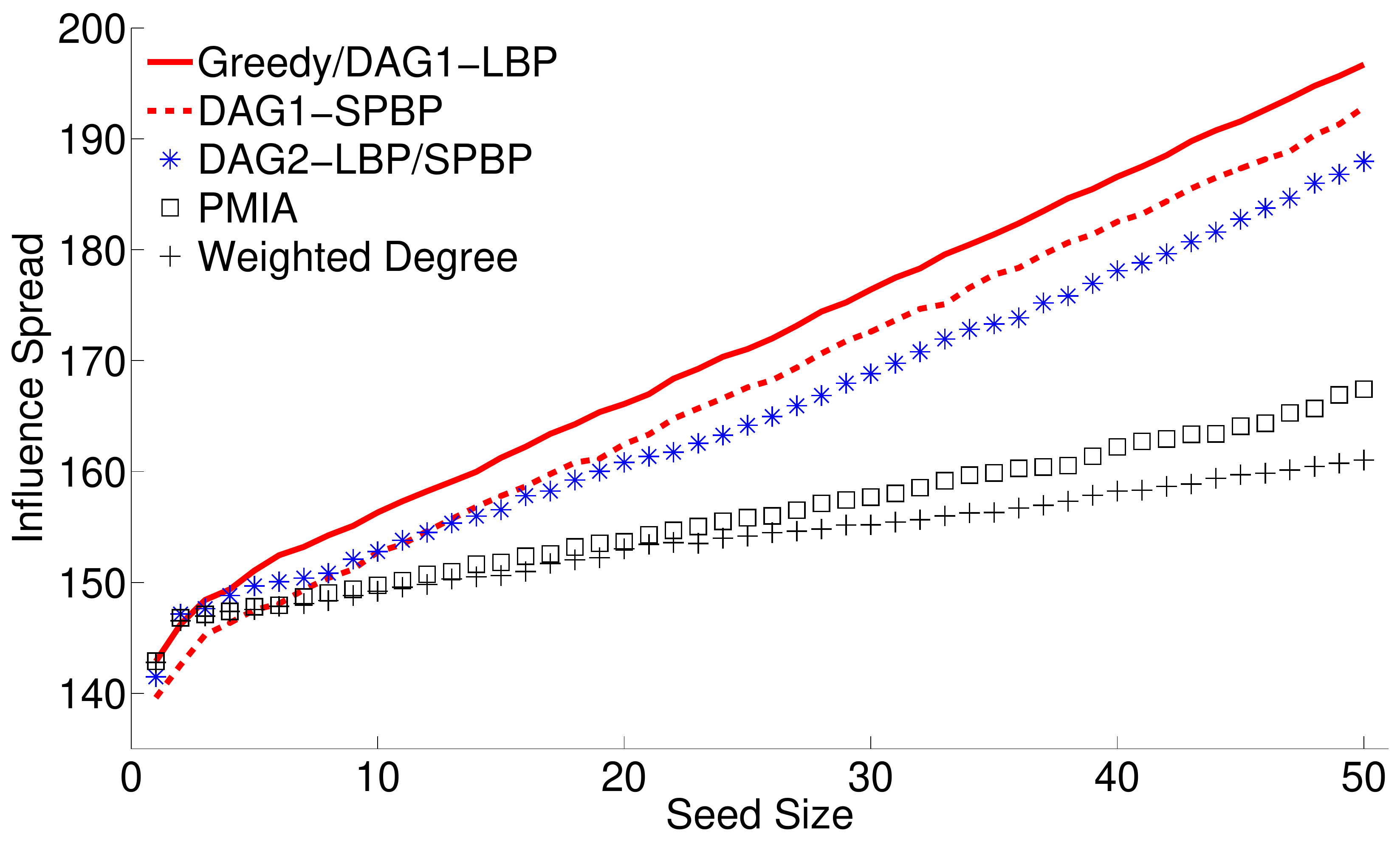} &
\includegraphics[width=1.75in]{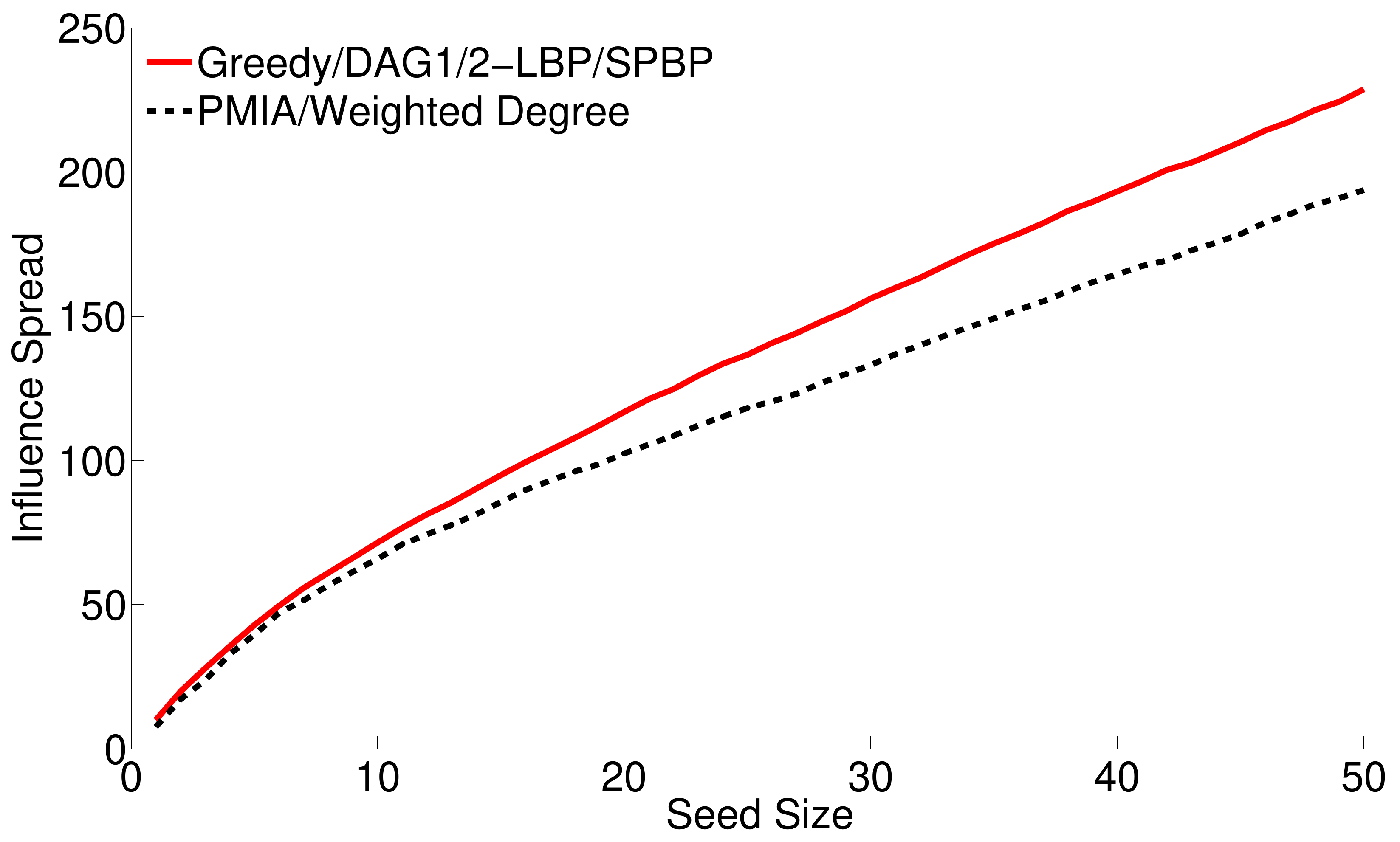} &
\includegraphics[width=1.75in]{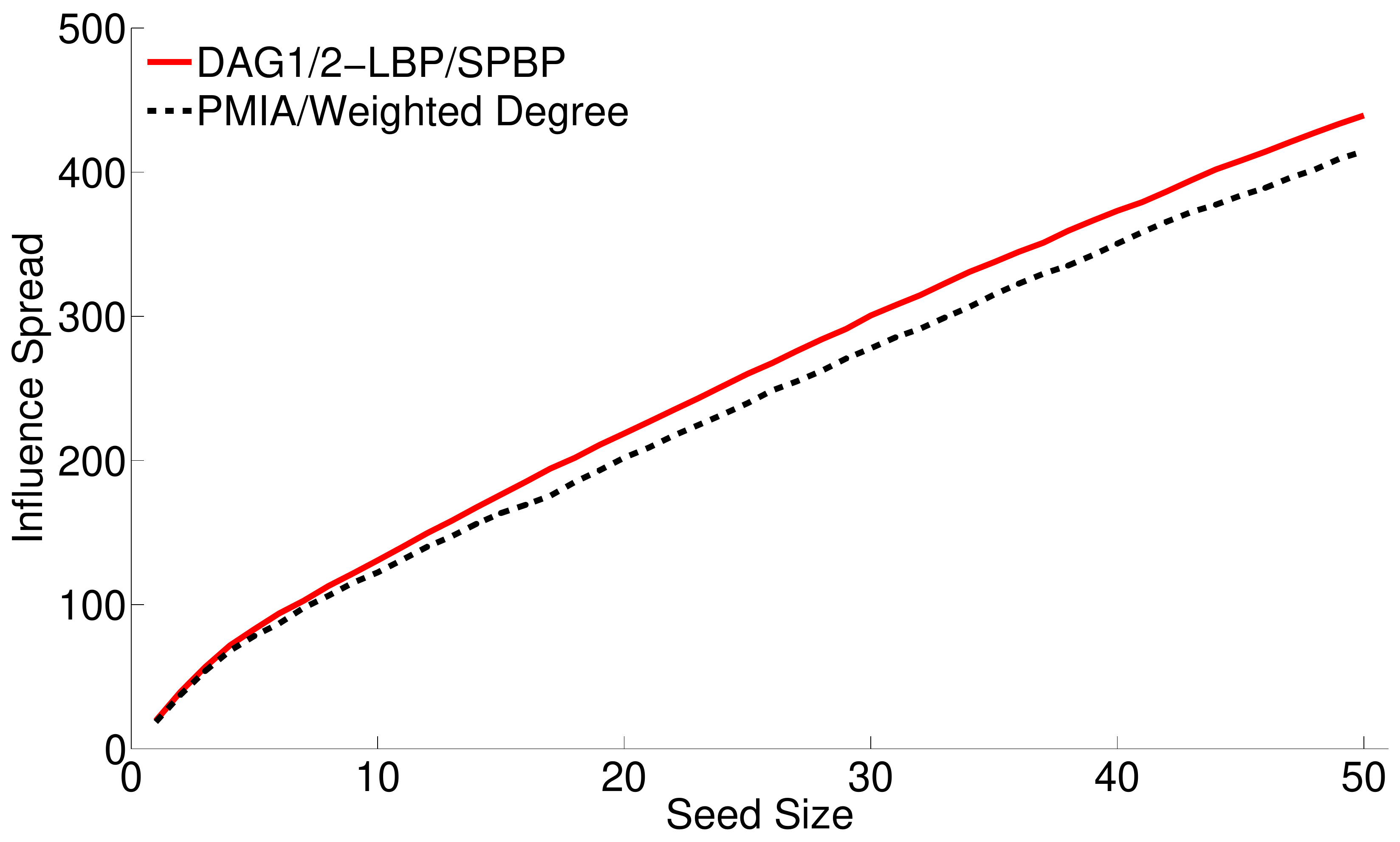} &
\includegraphics[width=1.75in]{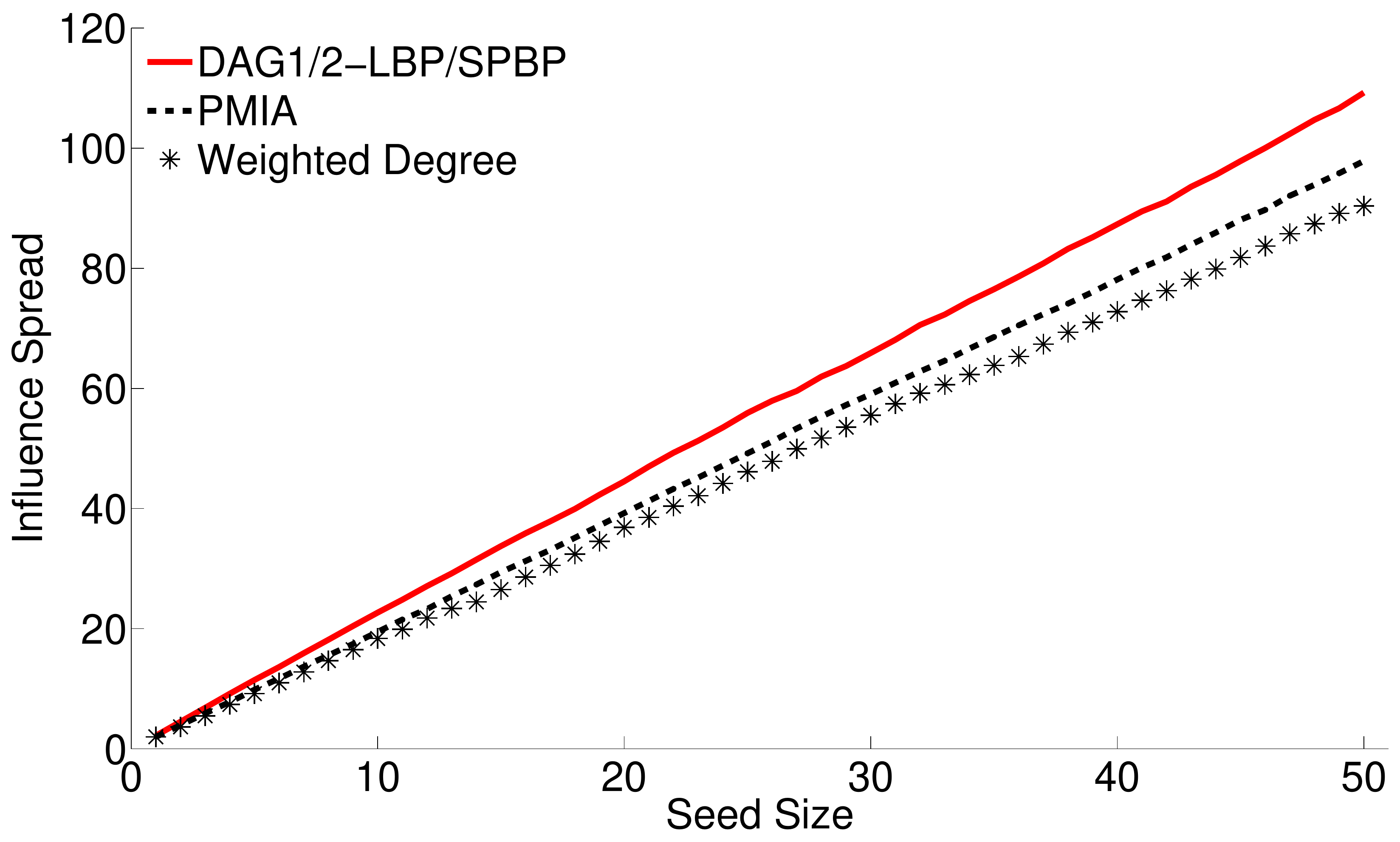} \\
(a) {\it Email} &
(b) {\it p2p-Gnutella} &
(c) {\it soc-Slashdot} &
(d) {\it Amazon}
\end{tabular}
\caption{Influence spread with node unit-cost on 4 datasets. DAG 1 results are in red curves, DAG 2 are in blue curves, and other methods are in black curves.}
\label{fig:influence_spread}
\end{figure*}

\begin{figure*}[t]
\hspace{-0.3in}
\begin{tabular}{cccc}
\includegraphics[width=1.75in]{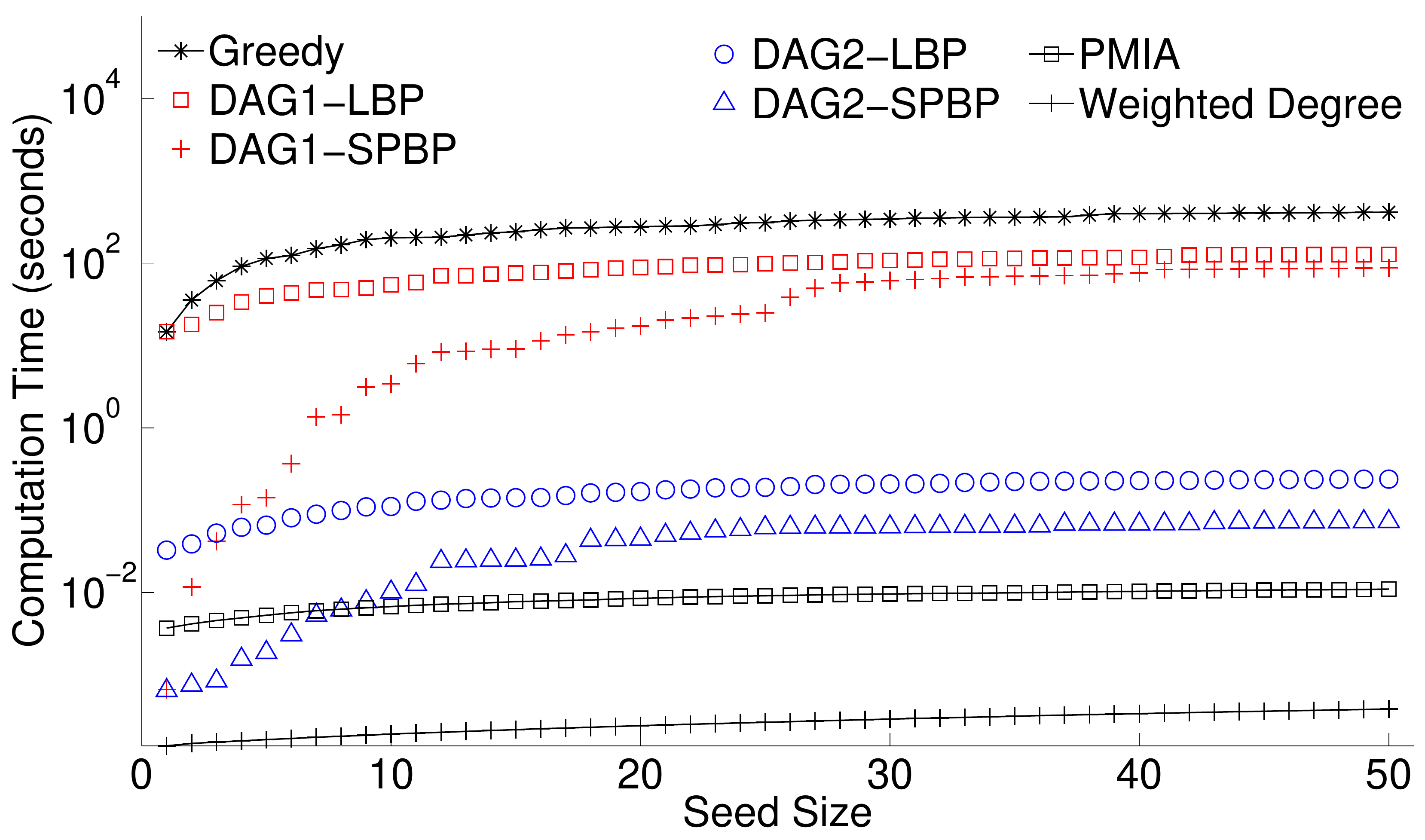} &
\includegraphics[width=1.75in]{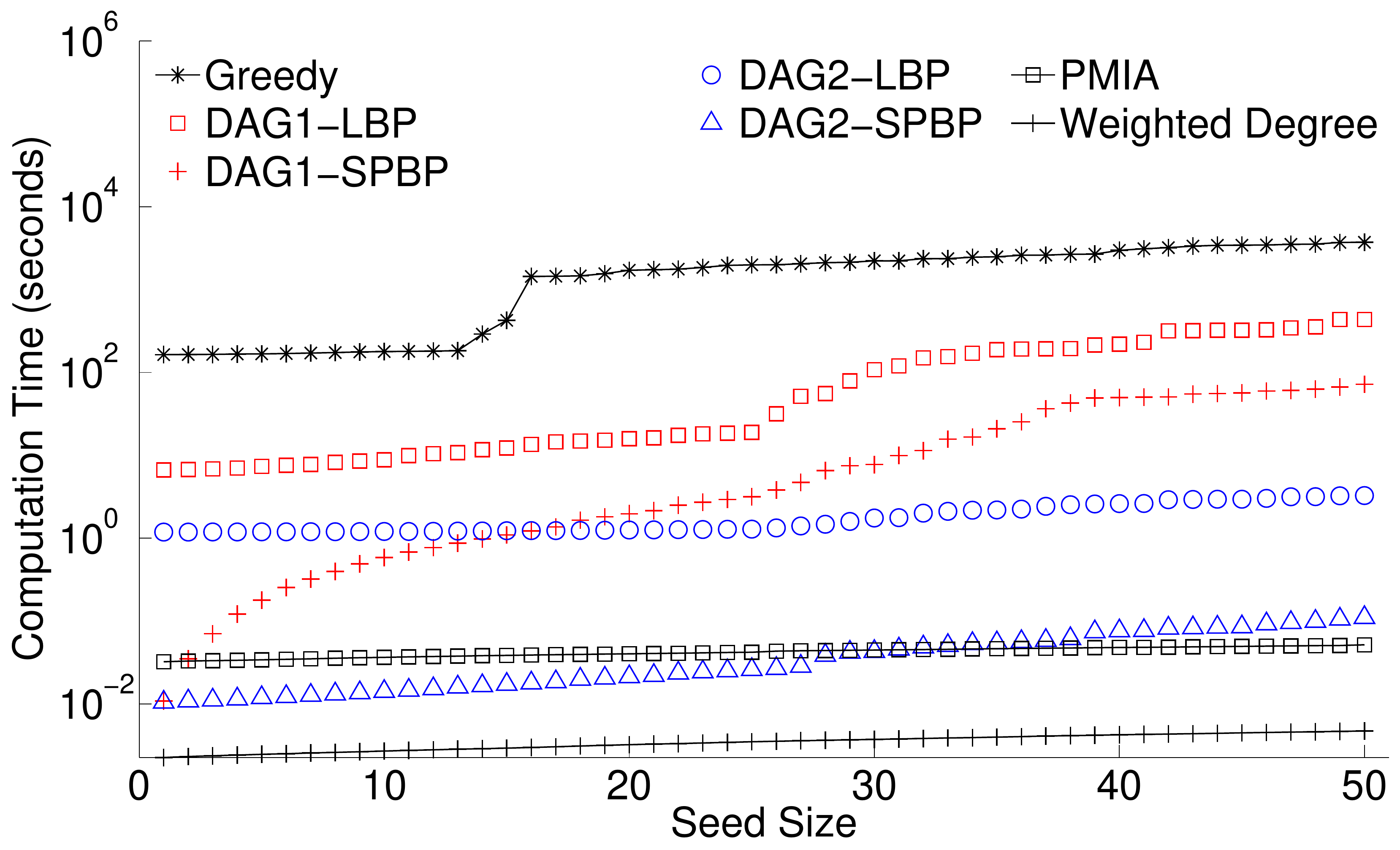} &
\includegraphics[width=1.75in]{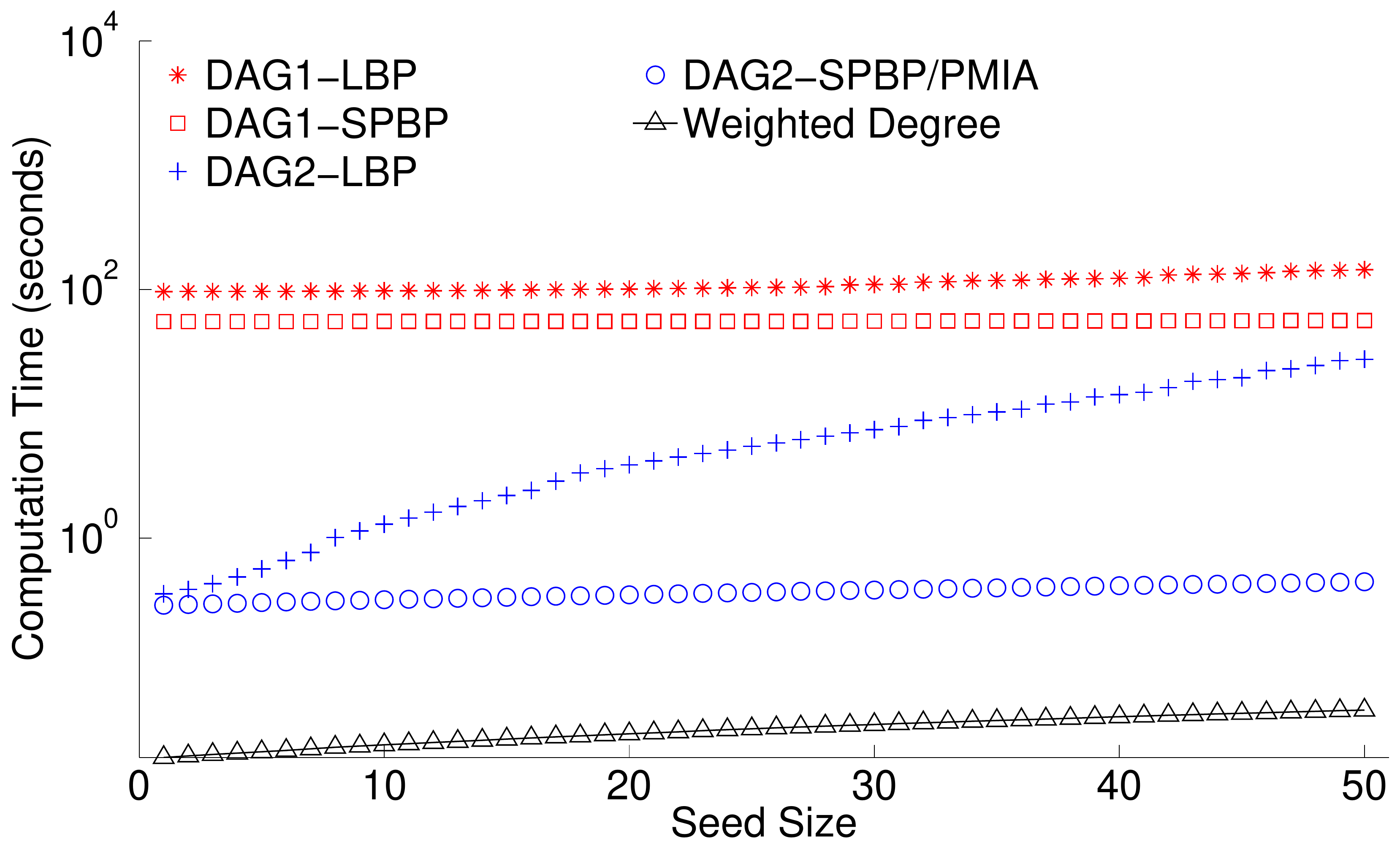} &
\includegraphics[width=1.75in]{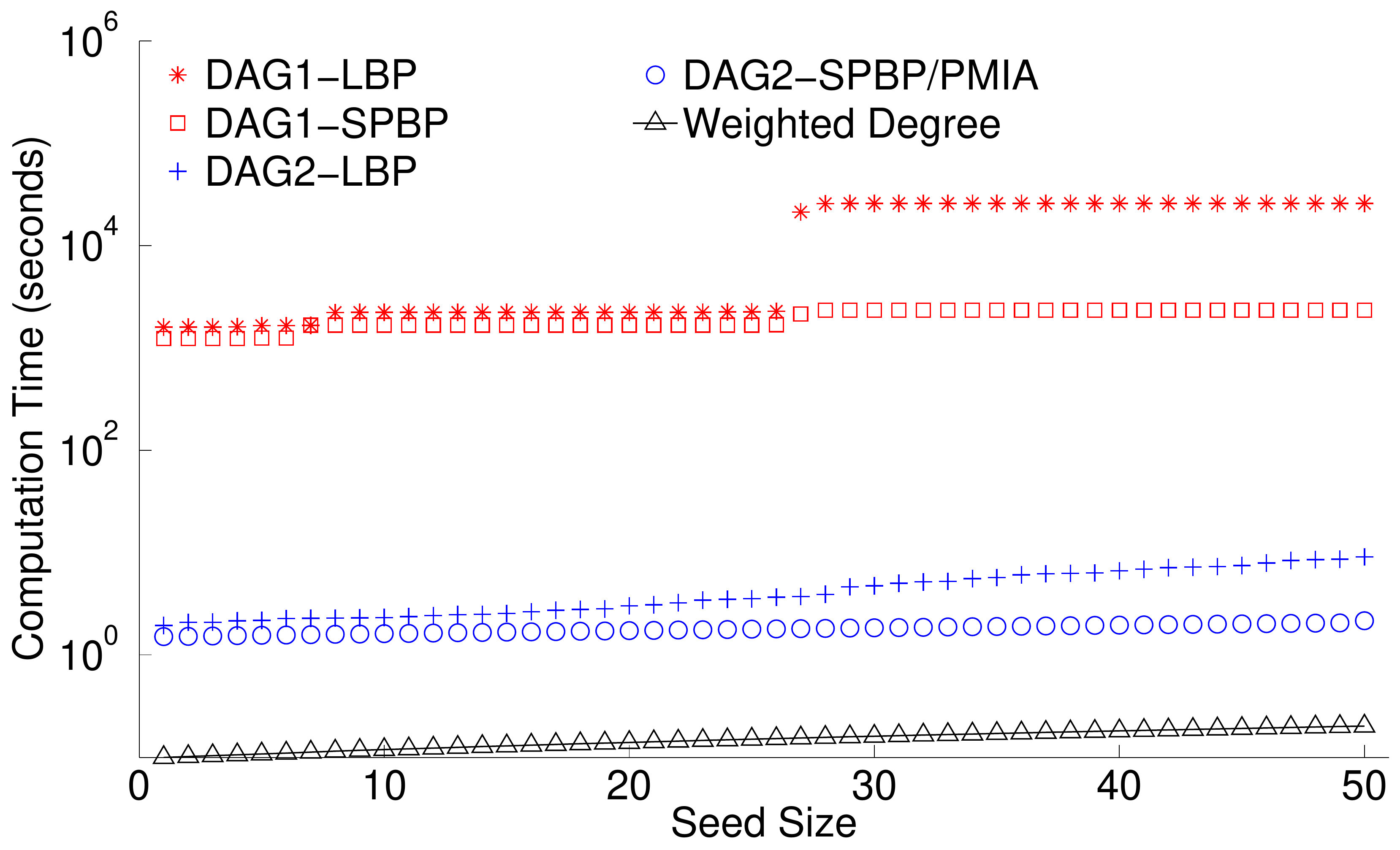} \\
(a) {\it Email} &
(b) {\it p2p-Gnutella} &
(c) {\it soc-Slashdot} &
(d) {\it Amazon}
\end{tabular}
\caption{Computation time with node unit-cost on 4 datasets.}
\label{fig:computation_time}
\end{figure*}

BIM with unit-cost is the traditional IM problem where the seed set size $k$ is fixed. In this experiment, we run 7 algorithms: Greedy, PMIA, Weighted Degree, and the 4 proposed methods on 4 datasets presented in Table II. $k$ varies from 1 to 50, and we adopt the RA probability generation model.

Figure~\ref{fig:influence_spread} shows the influence spread generated by the
best seed sets in different algorithms as the seed size changes. Since Greedy
does not scale with large datasets, we only run Greedy on {\it Email} and {\it
p2p-Gnutella}. The influence spread from the seed
set selected by each algorithm is determined by 10,000 rounds of Monte Carlo
simulations on the original graphs.

In Figure~\ref{fig:influence_spread}(a), the performance of DAG1--LBP and
Greedy (known to be within a constant ratio of the optimal) are not
distinguishable (and thus are represented in one curve). The influence spread
of DAG1--SPBP and DAG2--LBP/SPBP are shortly behind, all
outperforming PMIA and Weighted Degree. We observe on {\it Email} dataset (a small but dense
network) that both the structure of
the DAG (DAG 1 vs. DAG 2) as well as the BP algorithm used (LBP
vs. SPBP) will affect performance of the proposed methods. In
contrast, as shown in Figure~\ref{fig:influence_spread}(b) -- (d), the
influence spreads of the four approaches DAG1/2--LBP/SPBP are identical for
sparser networks, and is the same as Greedy in {\it p2p-Gnutella} dataset.

In terms of running time,  Weighted Degree is the fastest. Among the four
proposed approaches, DAG2--SPBP is the fastest, next are DAG2--LBP, DAG1--SPBP,
and finally DAG1--LBP. DAG2--SPBP and PMIA have comparable order in running time
with DAG2--SPBP being 30-40\% slower than PMIA in most cases. Again, this may be
primarily attributed to the lack of code optimization in our proposed methods.

Interestingly, influence spread on {\it Amazon} grows linearly with the seed size.
Our result matches with that in~\cite{pmia}. This can be explained by the sheer
scale of the network, and thus the small number of selected seeds are likely to
have non-overlapping influence regions.

\paragraph*{General cost version of BIM}

\begin{figure*}[t]
\hspace{-0.3in}
\begin{tabular}{cccc}
\includegraphics[width=1.75in]{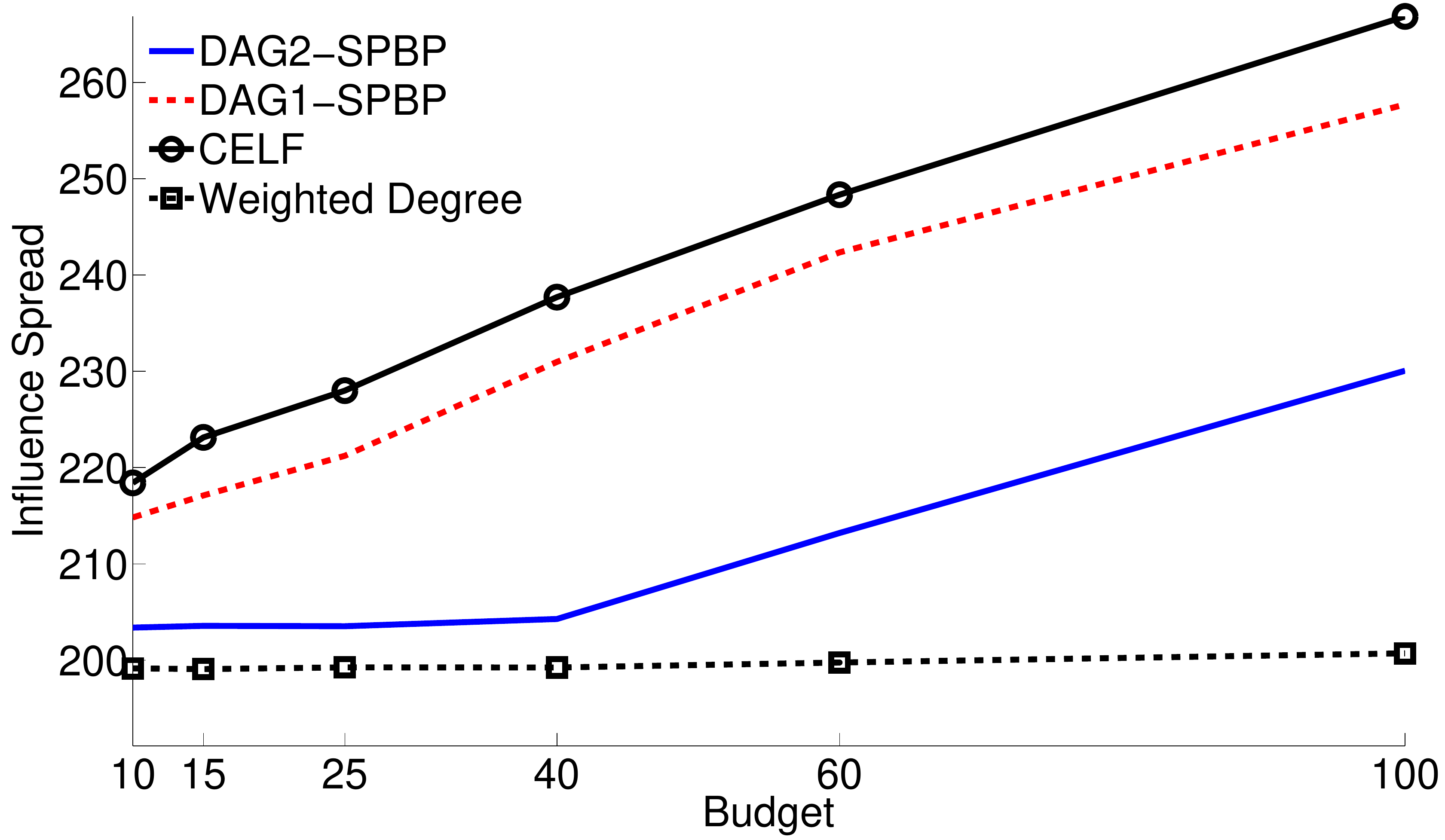} &
\includegraphics[width=1.75in]{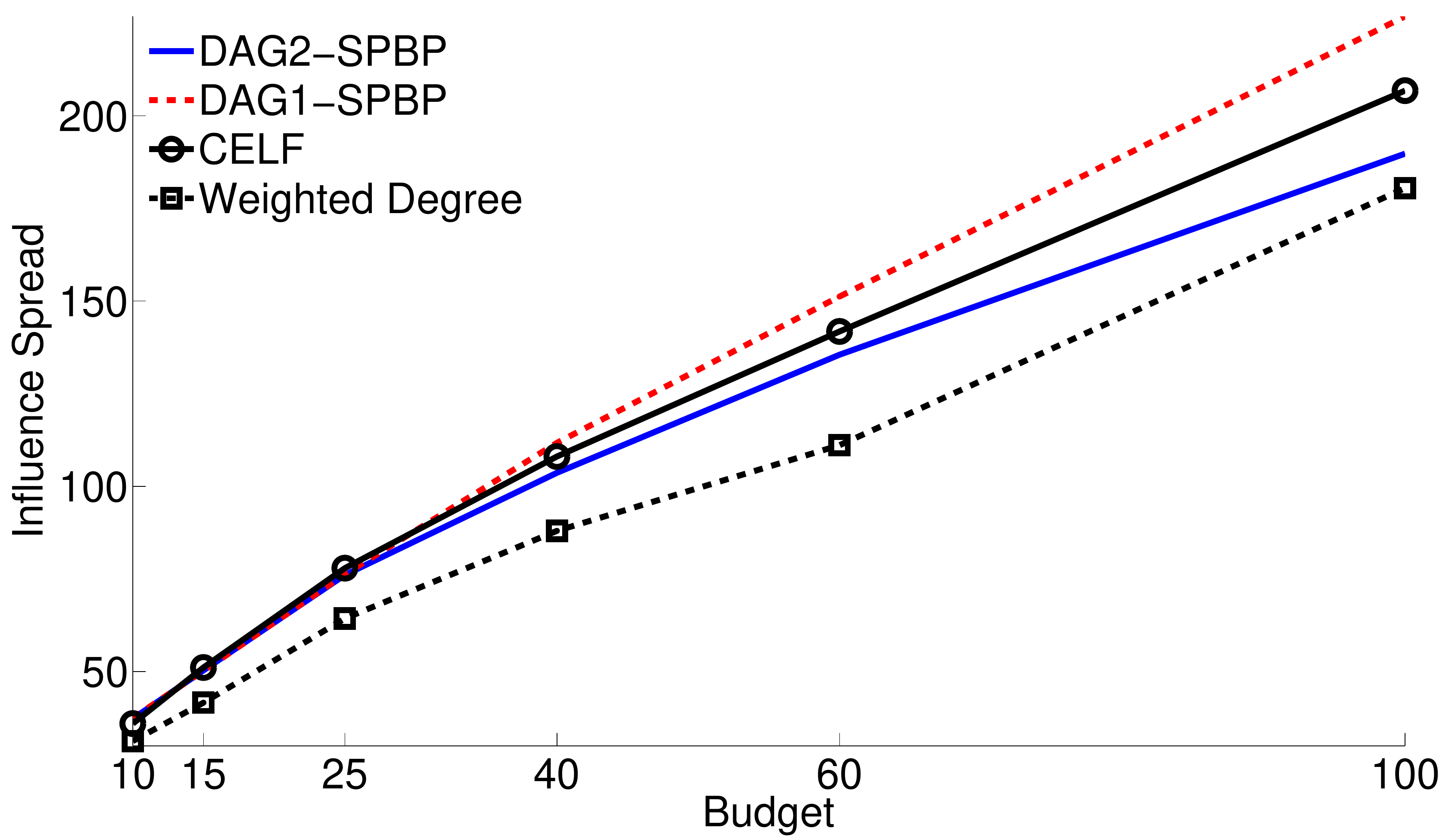} &
\includegraphics[width=1.75in]{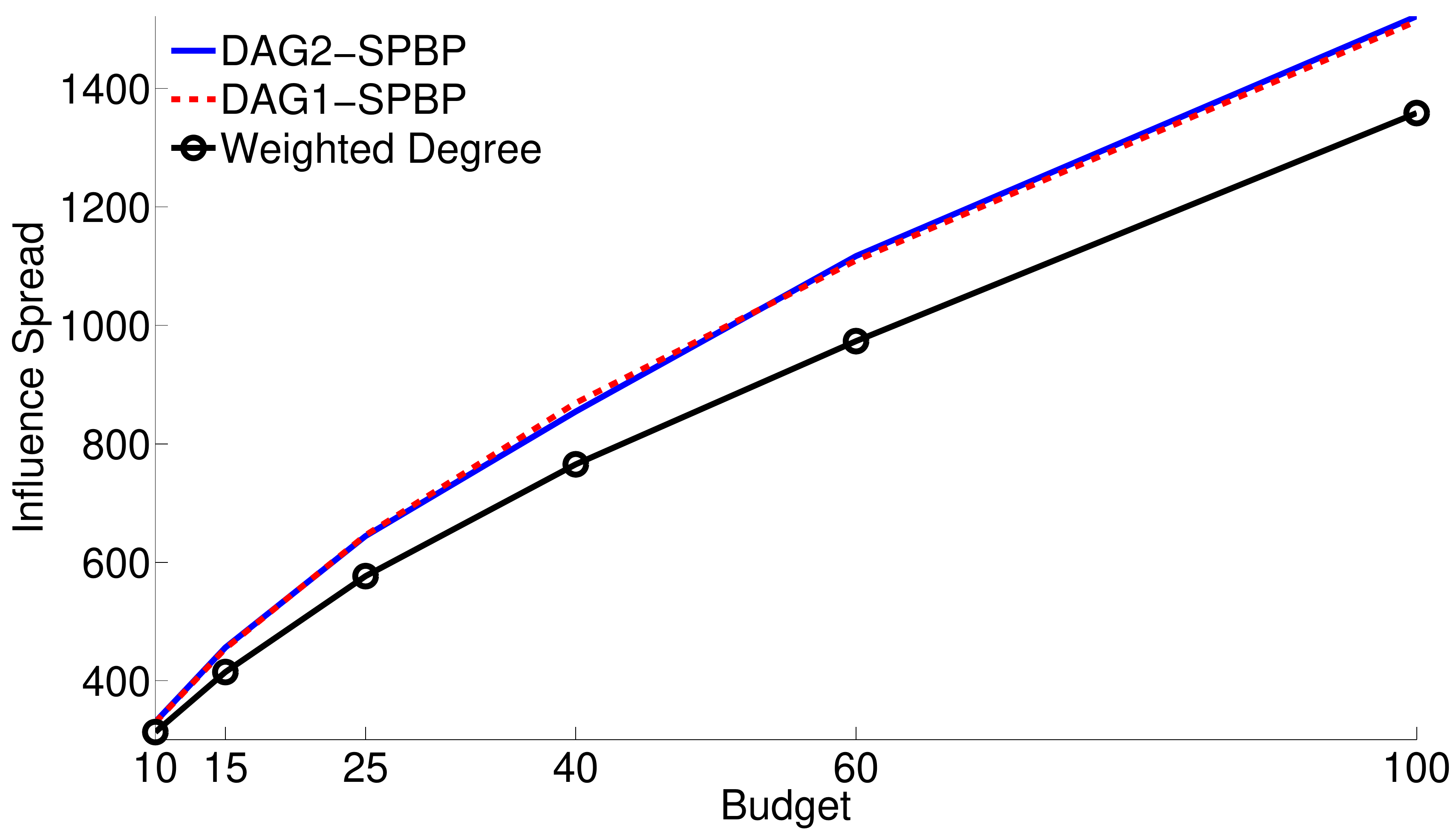} &
\includegraphics[width=1.75in]{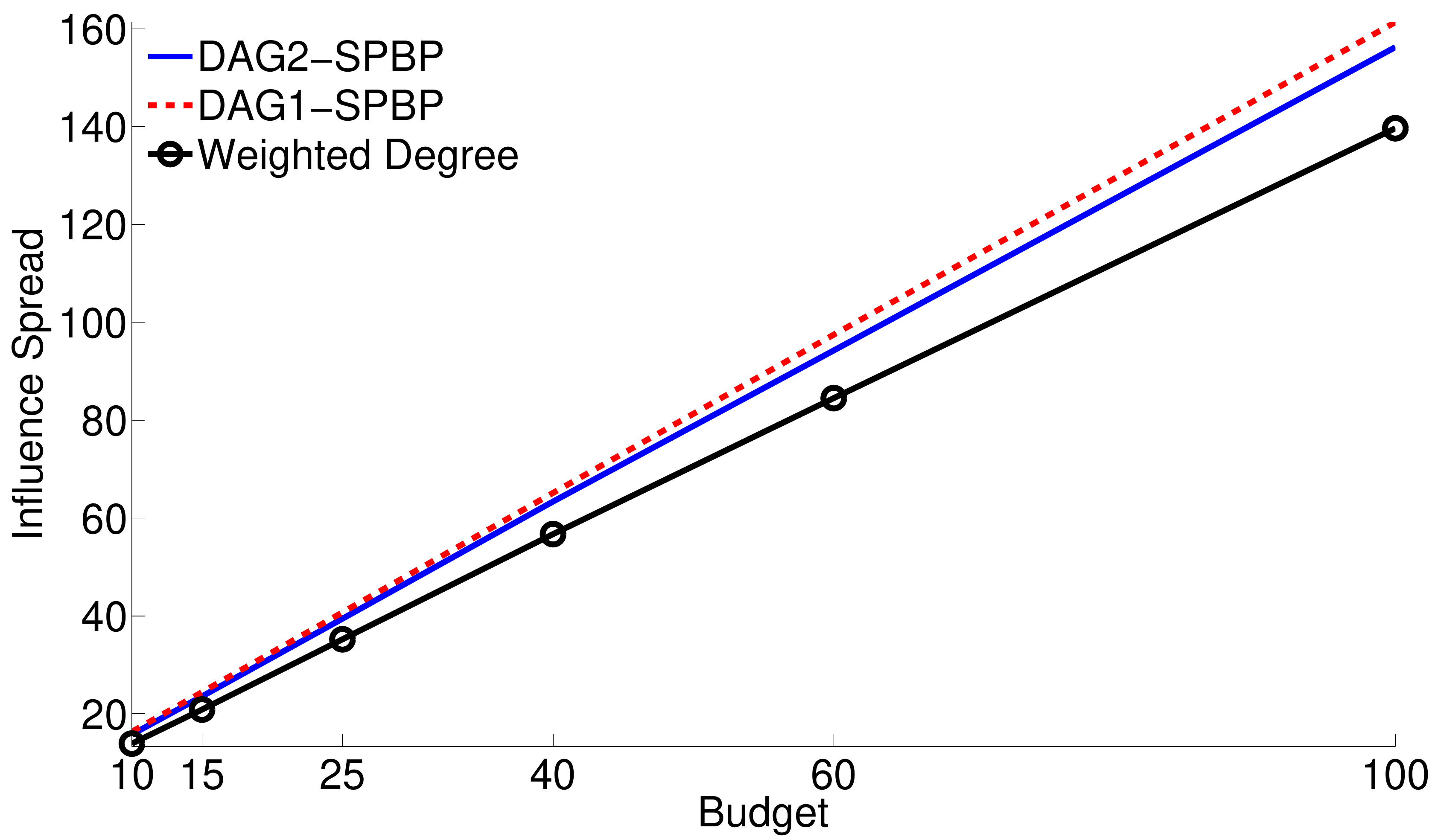} \\
(a) {\it Email} &
(b) {\it p2p-Gnutella} &
(c) {\it soc-Slashdot} &
(d) {\it Amazon}
\end{tabular}
\caption{Influence spread with random node costs on 4 datasets.}
\label{fig:BIM_influence_spread}
\end{figure*}

\begin{figure*}[t]
\hspace{-0.3in}
\begin{tabular}{cccc}
\includegraphics[width=1.75in]{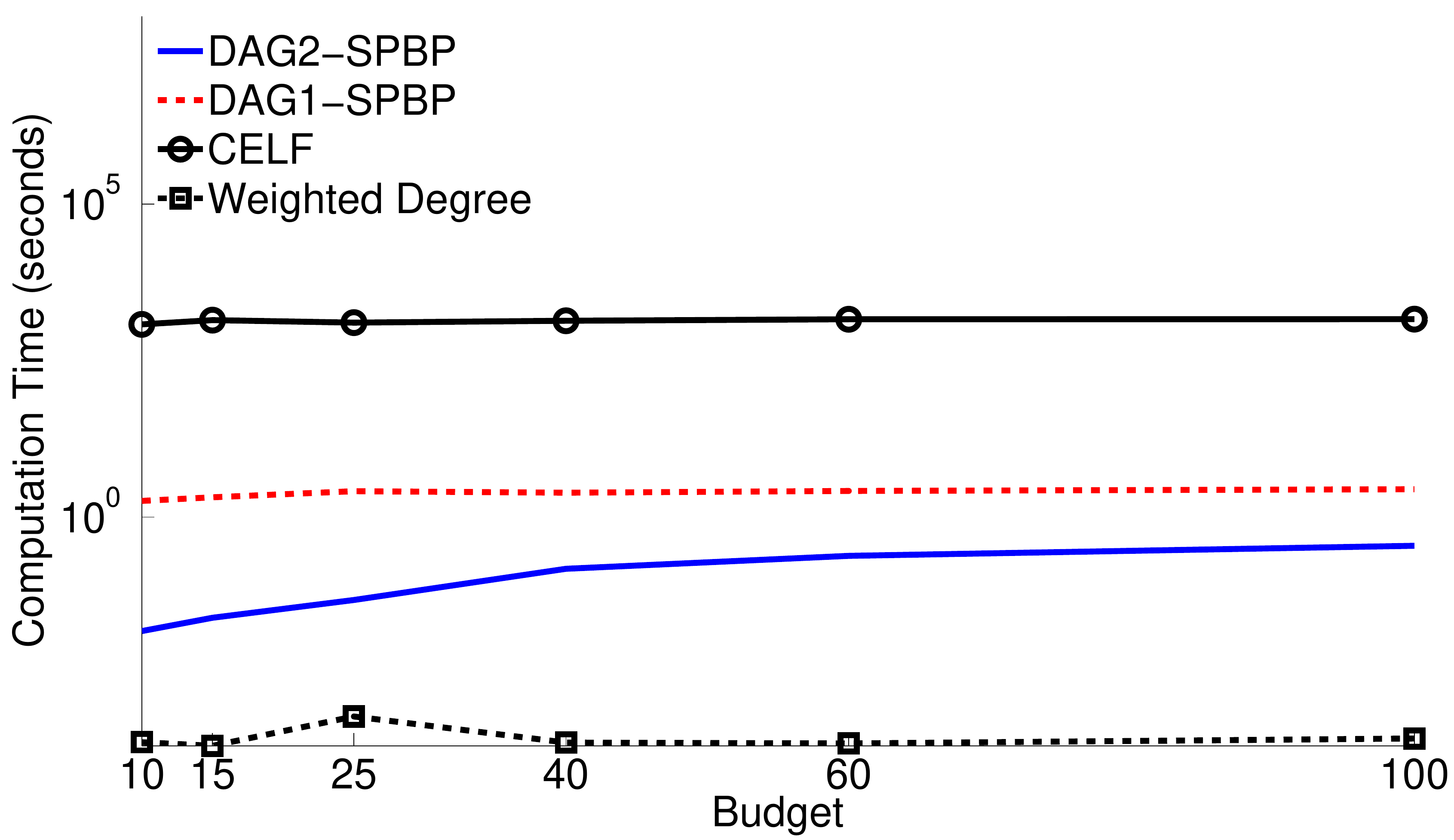} &
\includegraphics[width=1.75in]{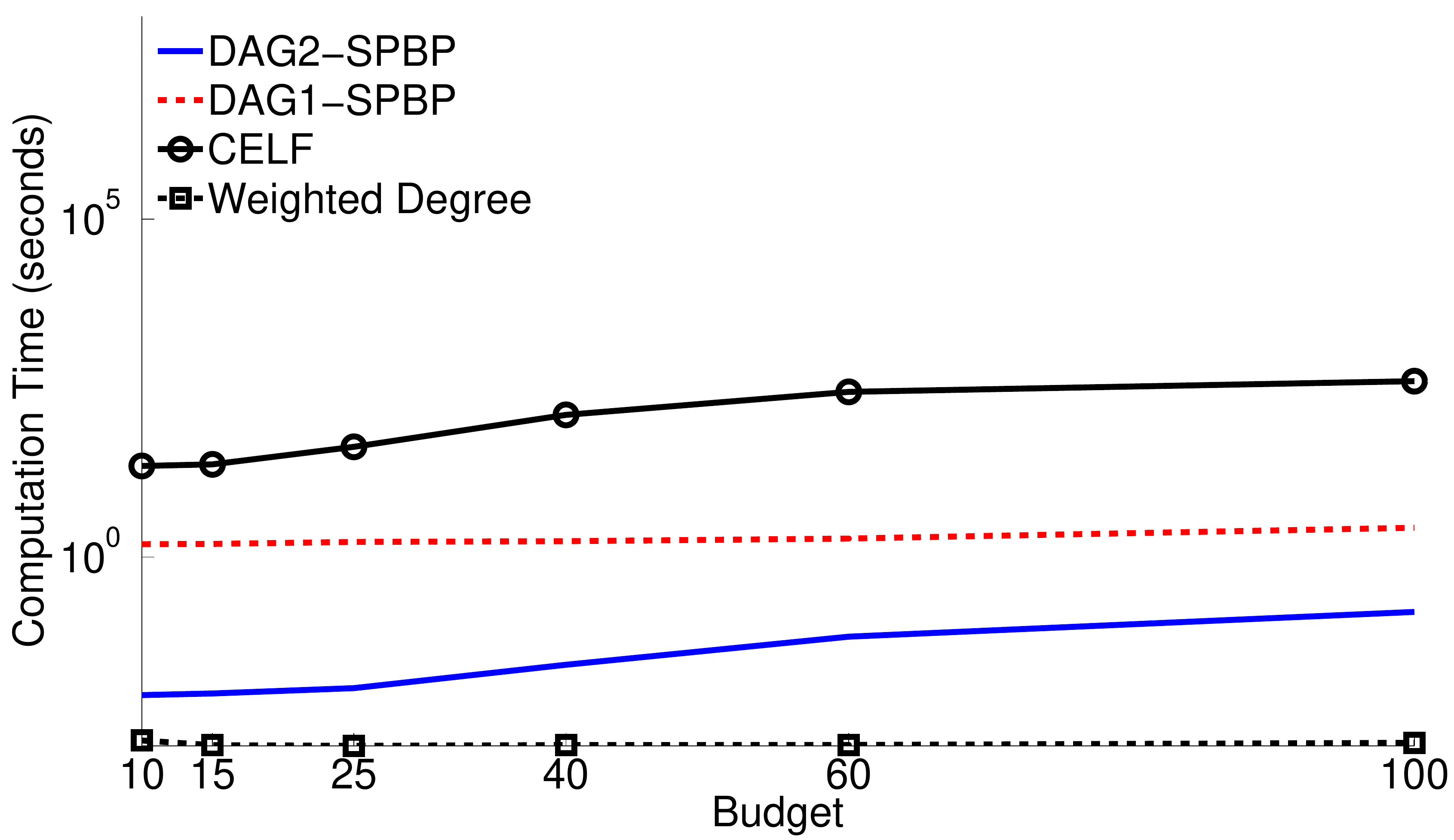} &
\includegraphics[width=1.75in]{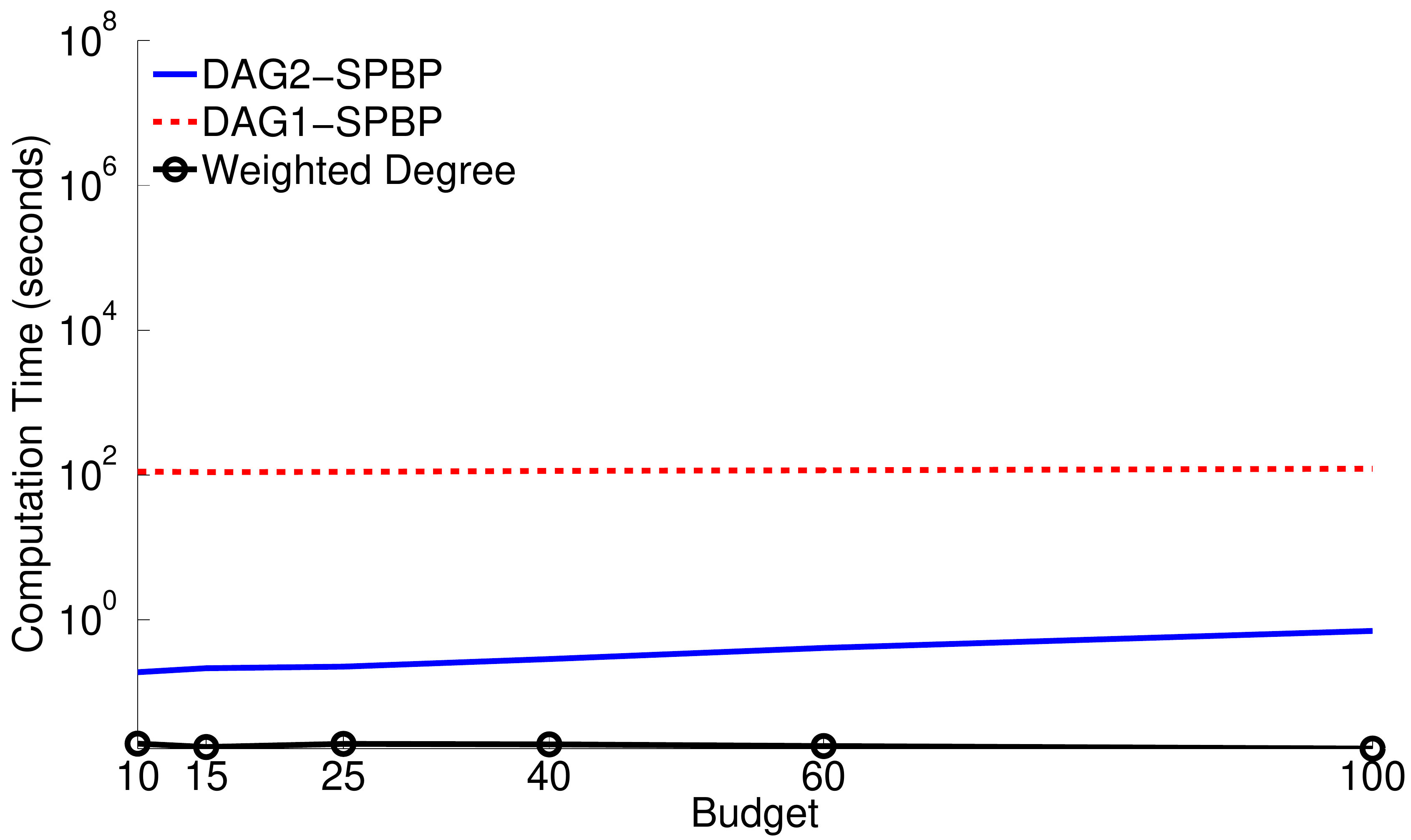} &
\includegraphics[width=1.75in]{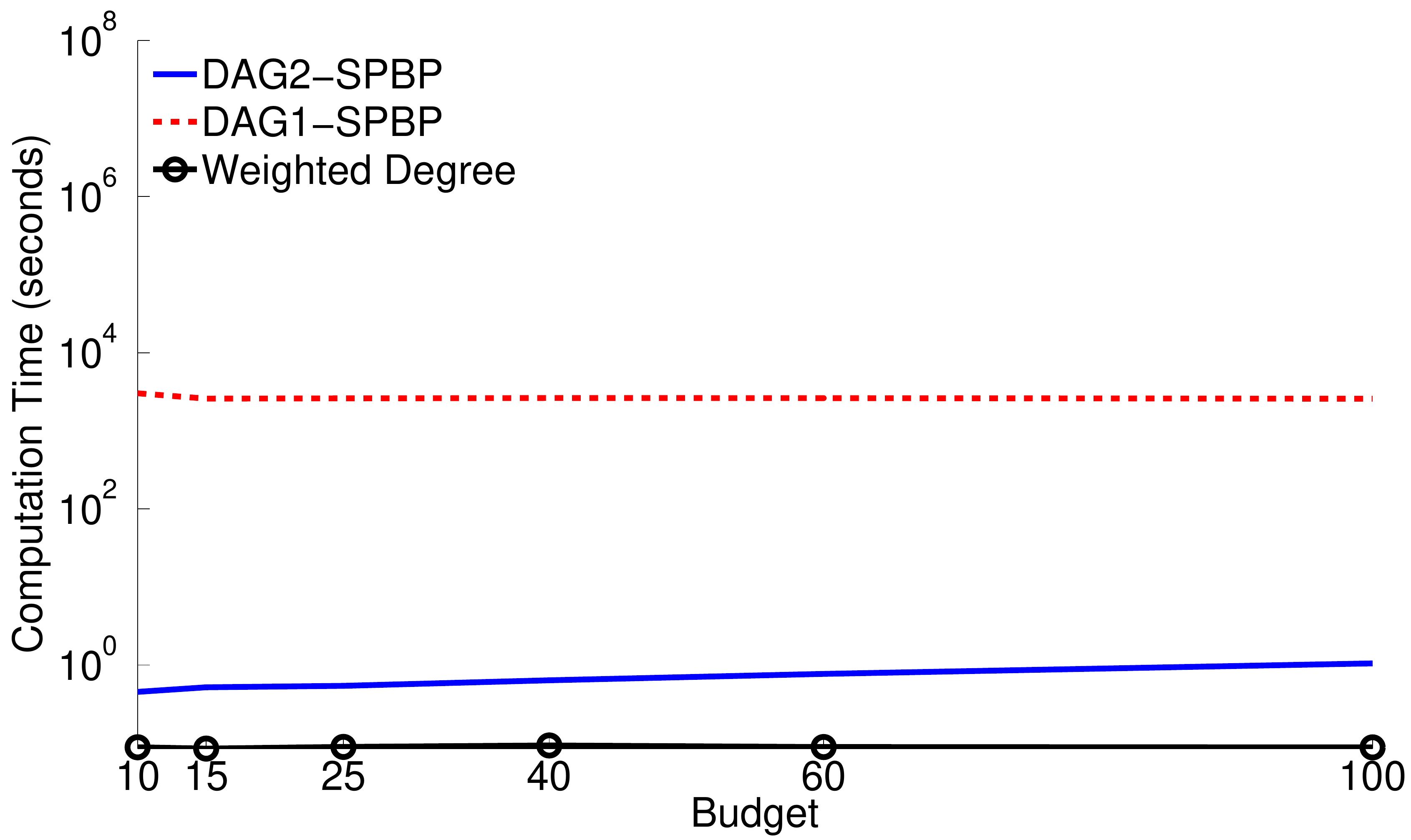} \\
(a) {\it Email} &
(b) {\it p2p-Gnutella} &
(c) {\it soc-Slashdot} &
(d) {\it Amazon}
\end{tabular}
\caption{Computation time with random node costs on 4 datasets.}
\label{fig:BIM_computation_time}
\end{figure*}

In this set of experiments, we compare only 4 algorithms: Greedy/CELF, Weighted Degree, and DAG1/DAG2--SPBP on 4 datasets presented in Table II.
We also omit the two methods that use LBP (DAG1/DAG2--LBP) from the result since they have comparable performance as the SPBP approaches. The budget $b = \{10, 15, 25, 40, 60, 100\}$, and the RA probability generation model is used. Nodal costs are selected uniformly in $[1.0, 3.0]$.

Results in Figures~\ref{fig:BIM_influence_spread} and~\ref{fig:BIM_computation_time} are similar to that in Figures~\ref{fig:influence_spread} and~\ref{fig:computation_time}. In most cases, DAG1 has better performance compared to DAG2. Notably, DAG1--SPBP outperforms Greedy/CELF on {\it p2p-Gnutella} dataset. Figure~\ref{fig:BIM_computation_time} shows that the proposed methods are several orders of magnitude faster than Greedy/CELF. Weighted Degree while being the fastest algorithm, does not perform nearly as well as the others on a dense graph ({\it Email}).

\paragraph*{Comparison of Influence Spread on Two DAG Models}

\begin{figure*}[t]
\hspace{-0.3in}
\begin{center}
\begin{tabular}{ c c }
\includegraphics[width=3.2in]{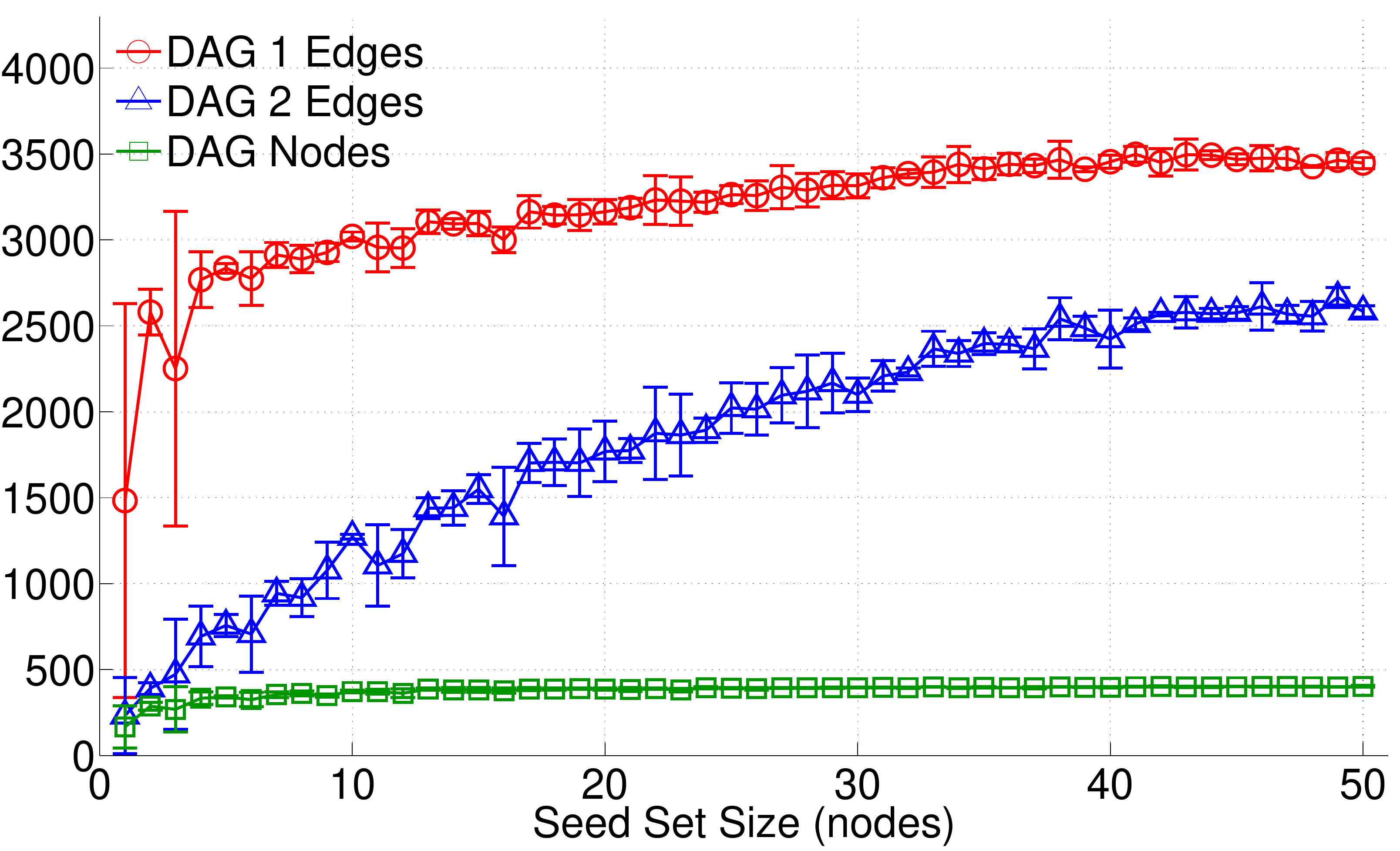}\hspace{0.2in} & \hspace{0.2in}
\includegraphics[width=3.2in]{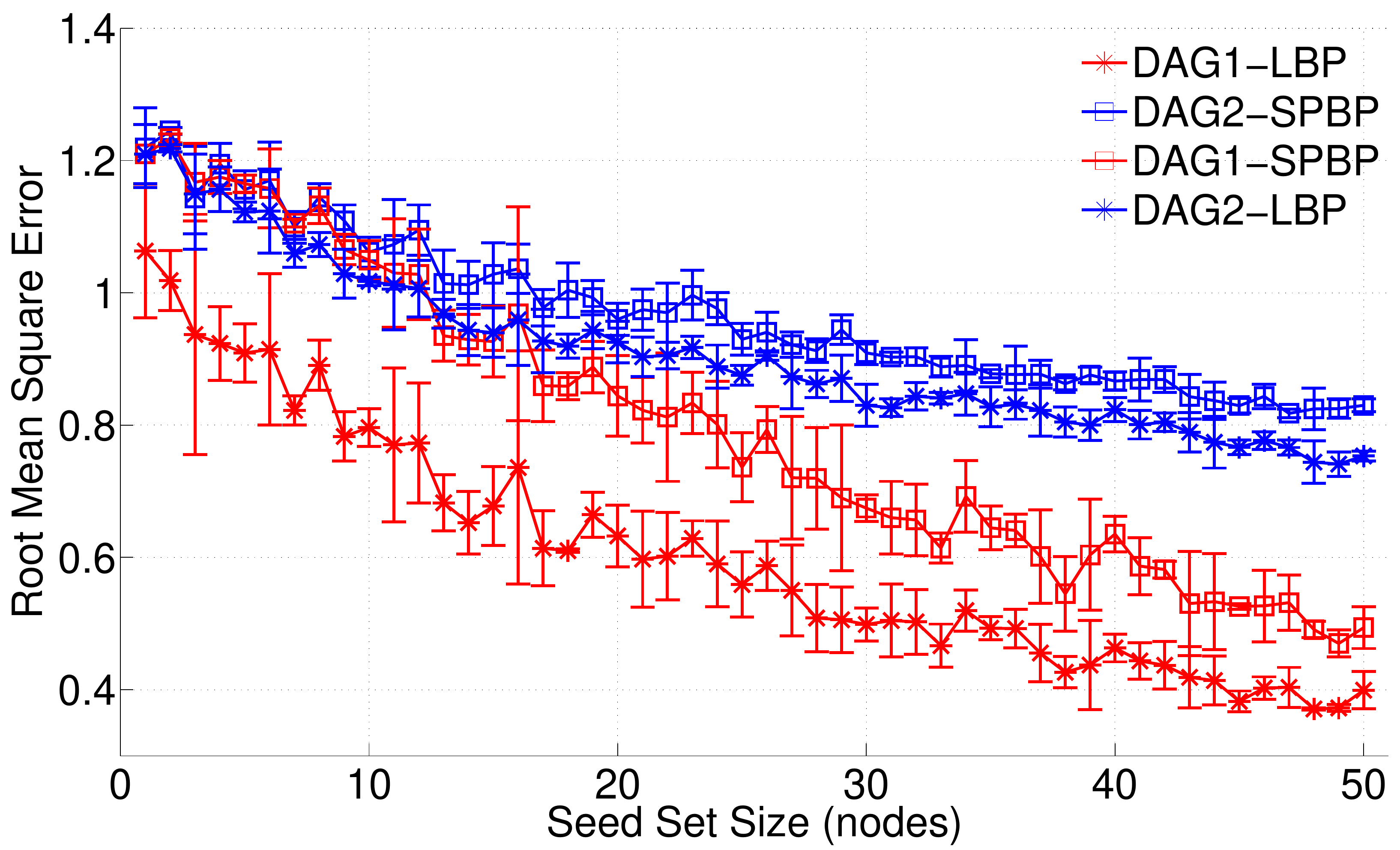}\\
(a) Number of nodes and edges in DAG &
(b) RMSE
\end{tabular}
\caption{Size of DAGs and RMSE of activation probabilities. Results are averages of 50 runs
with different seed selections and symmetric error bars indicate standard
deviations.}
\label{fig:BP1vsBP2}
\end{center}
\end{figure*}

To understand the behavior of the proposed algorithms, we conduct further
experiments on {\it Email} dataset as it gives the largest performance difference
between the algorithms.

Figure~\ref{fig:BP1vsBP2}(a) gives the number of vertices and edges as the
result of the two DAG models with varying sizes of seed sets. Since both have
the same number of vertices, only one curve is shown. It is clear that
DAG1 is much denser than DAG2 due to the inclusion of more edges. As the seed
set grows, the gap becomes smaller.

We use Root Mean Square Error (RMSE) to compare the activation probabilities on
nodes. RMSE is defined as,
$$
RMSE(p,p') = {\sqrt{\frac{\sum_{\forall v \in V}(p'(v)-p(v))^2}{n}}}/{\frac{\sum_{\forall v \in V}^{n}p(v)}{n}},
$$
where $p'(\cdot)$ is the inferred result from the propose algorithms. The
ground truth $p(\cdot)$ is determined by Monte Carlo simulations. When
$p'(v) = p(v), \forall v \in V$ then $RMSE(p,p') = 0$.

Figure~\ref{fig:BP1vsBP2}(b) shows that
DAG1 has smaller RMSE since it constructs a denser graph. More edges
clearly improves the quality of the seed selection process. In the comparing LBP and
SPBP, LBP is slightly better since SPBP ignores the correlation among node states.
The combination of DAG1 and LBP yields the best inference result, but incurs higher computation complexity.
The results are consistent with those in Figure~\ref{fig:influence_spread}(a).

\subsection{Synthetic Networks}

In this section, we conduct three sets of experiment with 5 methods: CELF, PMIA, Weighted Degree and DAG1/2--SPBP. Synthetically generated networks are used to study the impact of network structures and probability generation models on performance of the algorithms. To isolate the effects of network properties, we only consider the unit cost BIM problem.

\paragraph*{Impact of network density}
Results from Figure~\ref{fig:influence_spread} and~\ref{fig:BIM_influence_spread} indicate that our proposed methods perform best on dense networks ({\it Email} and {\it p2p-Guntella}). To further validate this observation, we generate 4 networks with 20k, 50k, 100k, and 200k edges using DIGG~\cite{digg}. The number of vertices is fixed at 5,000. Seed set size $k = 50$ and probability model is RA. We evaluate the spread ratio of various algorithms, defined as the ratio of the spread attained to that by Greedy/CELF algorithm. From Figure~\ref{fig:variousenvironment}(a), as the network density increases, the performance gap between the proposed algorithms and existing algorithms including CELF increases. CELF relies on many rounds of simulations to determine the spread. For dense networks, more rounds of simulations are needed to produce a spread estimation that is close enough to the {\it ground truth}. As a result, with a fixed number of simulation rounds, CELF has worse performance at high network densities. We also observe that PMIA, which was designed to take advantage of network sparsity; and Weighted Degree, which only uses local node information, do not perform well on densely connected graphs.

\begin{figure*}[t]
\hspace{-0.3in}
\begin{center}
\begin{tabular}{ c c c }
\includegraphics[width=2in]{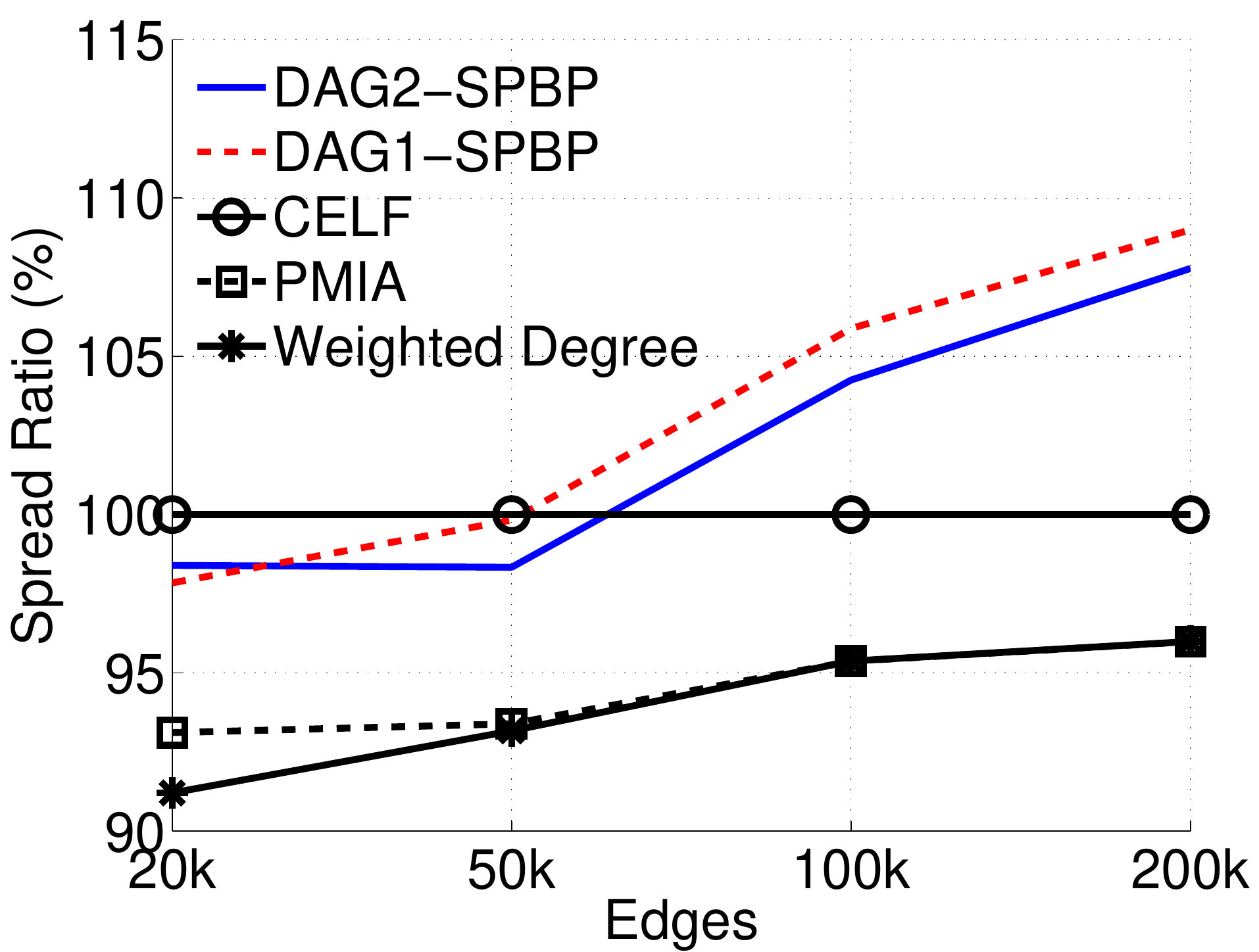} &
\includegraphics[width=2in]{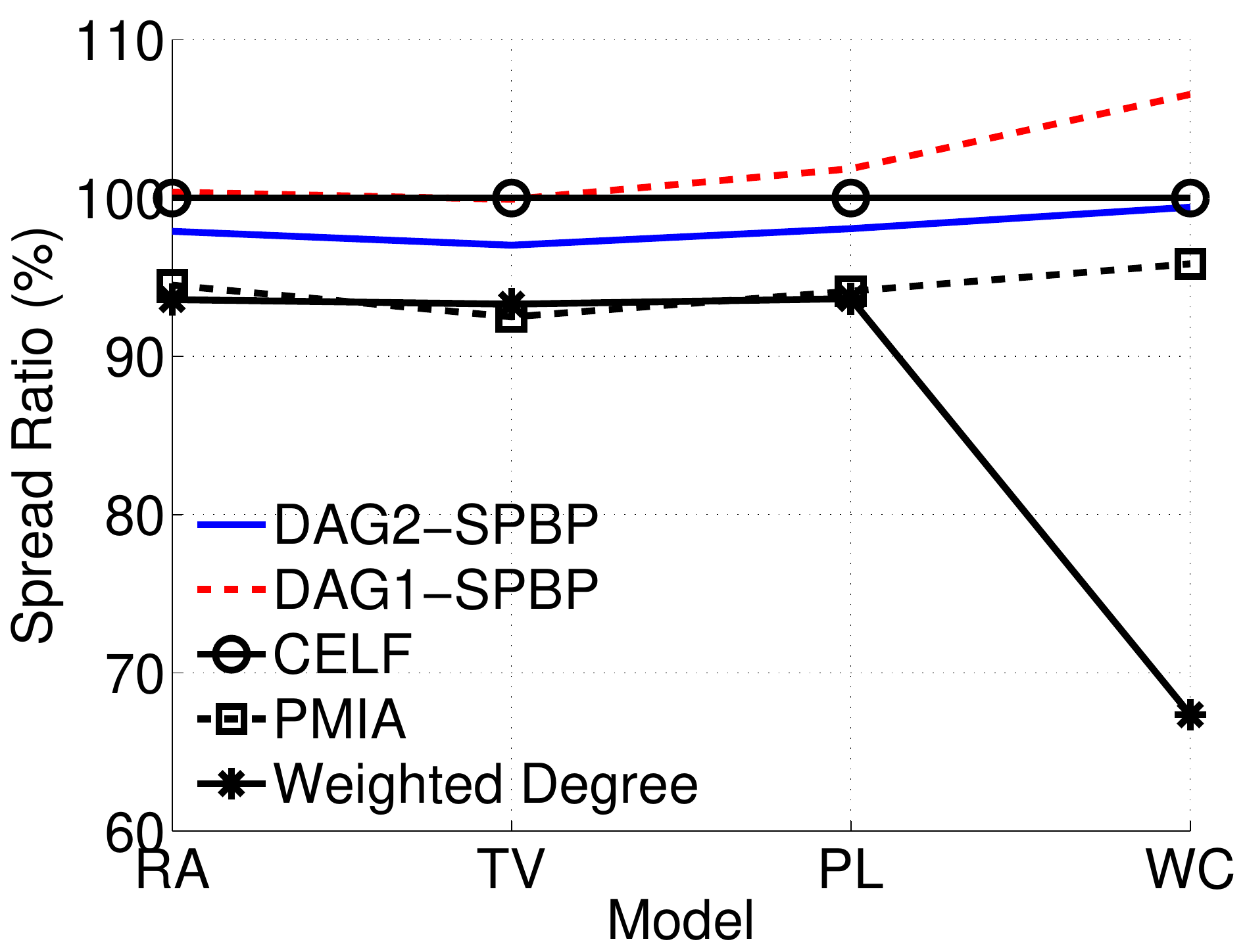} &
\includegraphics[width=2in]{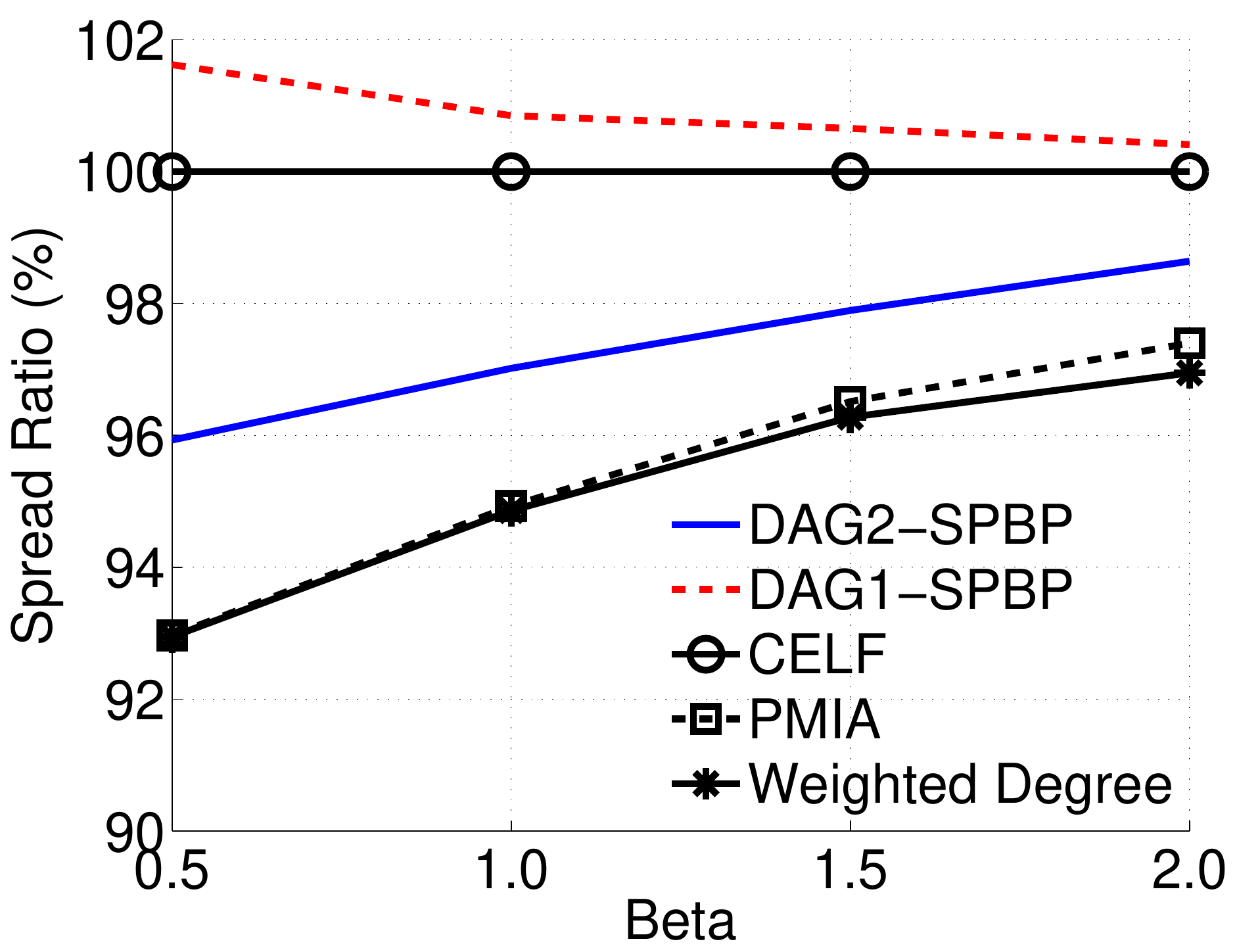}\\
(a) Varying network densities &
(b) Varying probability generation models &
(c) Varying out-degree distributions
\end{tabular}
\caption{Algorithm performance on different network conditions.}
\label{fig:variousenvironment}
\end{center}
\end{figure*}

\paragraph*{Impact of probability generation model}
In this set of experiments, we run 5 algorithms on a synthetic network with 5,000 nodes and 50,000 edges. Each algorithm selects a seed set with size $k = 50$ under 4 propagation probability models: RA, TV, PL and WC. All models give similar performance except Weighted Degree on WC model. Recall that WC generates the propagation probabilities based on the in-degree of nodes, thus strong connections are established between nodes with low in-degree. Weighted Degree can't ``see'' those strong ties beyond the local edges, and therefore, has the worst performance.

\paragraph*{Impact of node out-degree distribution}
It is known that node out-degree in real social networks follows the power-law distribution~\cite{powerlaw}. Let $y$ be the percentage of nodes with degree $x$, then we have $y \sim \alpha / x^{\beta}$.$\alpha$ and $\beta$ can be seen as the intercept and the (negative) slope when degree sequence is plotted on a log-log scale. While varying $\alpha$ only scales the distribution up or down, changing $\beta$ alters the ``shape'' of the distribution. More specifically, a high value of $\beta$ means the node out-degree distribution exhibits larger skew. The network in this case contains few ``hubs'' that are connected to many other nodes. On the other hand, a small $\beta$ means that the distribution is fat-tailed and the max out-degree in the network is not much larger than the average out-degree. We run 5 algorithms to solve the unit-cost BIM problem on 4 generated networks with $\beta = \{0.5, 1.0, 1.5, 2.0\}$. The network size is 5,000 and $\alpha$ is adjusted accordingly such that the total number of edges is roughly 50,000. We see from Figure~\ref{fig:variousenvironment}(c) that the performance gap among the algorithms reduces with larger $\beta$. This is because with a large degree distribution skewness, nodes with high out-degree (hub) will almost certainly be one of the best seed candidates (unless their costs are too high, which is not this case). Simple algorithms such as Weighted Degree can easily identify such hub nodes. On the other hand, when the network is more ``flattened'', more sophisticated algorithms are necessary.

\subsection{Summary}
From the experiments results, Weighted Degree gives the best efficiency
in terms of spread/complexity. However, its performance degrades significantly on dense networks or more heavy tailed power law graphs. The same conclusion is applied to PMIA. Even though being faster than our algorithms, PMIA shows little improvement in term of attainable spread compared to Weighted Degree, except under the WC model.
Our proposed schemes surpass the others in all the experimented
datasets. They also offer more application flexibility: one would apply the best
performed algorithm (DAG1--LBP) on static networks (e.g.: network of connections between
co-workers) to identify the most influential nodes, or apply the fastest algorithm (DAG2--SPBP)
on rapidly changing communities (e.g.: network of connections between people in a social group)
to obtain immediate result.


\section{Conclusion}
\label{sec:conclusion}
While recent researches focus on solving the IM problem, we considered in this paper the BIM problem, which is a generalization of the former one.
The study on real world datasets and synthetic datasets with controllable network parameters provides convincing evidences the proposed algorithms have superior performance. Furthermore, we gain some insights on the choice of algorithms in trading computation complexity with performance given the network structure.

\section{Acknowledgment}
This work is supported in part by the National Science Foundation under grants CNS-1117560 and CNS-0832084.

\bibliographystyle{IEEEtran}

\appendix
\label{app:sharphard}
\paragraph*{Proof of Theorem 1}\mbox{}
\begin{proof}
The proof is an adaption of the proof in \cite{pmia} and Valiant's original proofs
of the \#P-completeness of the $s$-$t$ connectedness in a direct graph~\cite{Valiant_1979}.
First, we define a few problems that are known or to be proven to be \#P-complete.
\begin{defn} (SAT')
\begin{itemize}
\item[Input:] $F = c_1\wedge c_2\wedge \ldots c_r$, where $c_i = (y_{i1}\vee y_{i2})$ and $y_{ij} \in X$,
\item[Output:] $|\{(\xx,\ttt)| \ttt = (t_1, t_2, \ldots, t_n) \in \{1,2\}^n$; for
$1\le i\le r$, $\xx$ make $y_{i,k}$ true for $k=t_i$.
\end{itemize}
\end{defn}
\begin{defn} (S-SET CONNECTEDNESS on DAG)
\begin{itemize}
\item[Input:] A DAG $\MD = (V,E); s \in V; V'\in V$.
\item[Output:] Number of subgraphs of $\MD$ in which for each $u\in V'$, there is a (directed) path from $s$ to $u$.
\end{itemize}
\end{defn}
\begin{defn} (S-T CONNECTEDNESS on DAG)
\begin{itemize}
\item[Input:] A DAG $\MD = (V,E); s, t\in V$.
\item[Output:] Number of subgraphs of $\MD$ in which there is a directed path from $s$ to $t$.
\end{itemize}
\end{defn}

To prove Theorem 1, we first establish the following lemma.
\begin{lem}
$SAT'\preceq_{p}$ S-T CONNECTEDNESS on DAG.
\label{lem:stc}
\end{lem}
\begin{proof}
Given $F$ construct a DAG $\MD = (V, E_1\cup E_2)$ where
$V = \{c_1, c_2, \ldots, c_{r+1}, x_1, \ldots, x_n, \bar{x}_1, \ldots,\bar{x}_n, s\},$
$E_1 = \{(x_i, c_j) | x_i \mbox{ appears in clause $c_j$ in $F\}\bigcup \\  \{(x_n,
c_{r+1}), (\bar{x}_n, c_{r+1})\}$},$ and
$E_2 = \{(x_i, x_{i+1}), (\bar{x}_i, x_{i+1}), (\bar{x}_i, \bar{x}_{i+1}), (x_i, \bar{x}_{i+1})|1\le i\le n\}\bigcup \{(s,x_1),(s,\bar{x}_1)\}.$
The direction of each edge follows the order of the pairs. $\MD$ is a DAG as edges only go from $x$'s of smaller index to
larger ones, and from $x$'s to $c$'s. Note the $\MD$ is multi-connected. The rest
of the proof follows that in \cite{Valiant_1979}.
\end{proof}
Theorem 1 can then be proved using the same argument as in
\cite{pmia} with the exception that the reduction is from the S-T CONNECTEDNESS on
DAGs due to Lemma~\ref{lem:stc}.
\end{proof}

\paragraph*{Proof of Proposition 1}\mbox{}
\begin{proof}
In both algorithms, a node $v$ is not included in the DAG if and only if $r(v)
> \theta$. Thus, $V_{\MD_1} = V_{\MD_2}$.

To show $E_{\MD_2} \subseteq E_{\MD_1}$, it suffices to show that $\forall (u,v) \in E_{\MD_2}$, $(u,v) \in E_{\MD_1}$. Since $(u,v) \in E_{\MD_2}$, $(u,v) \in E$ and $r(u) \le r(v)$. Therefore, according to Algorithm 2, $(u,v) \in E_{\MD_2}$. Clearly, the converse is not true as some edges in $E_{\MD_1}$ may not be part of the $MIOA$ from any seed node.
\end{proof}

\paragraph*{Proof of Proposition 2}\mbox{}
\begin{proof}
It is easy to see that by limiting the spread from $u$ in $MIOA(\MG,u,\theta)$,
then $p(w), \forall w \in MIOA(\MG,v,\theta)$ will not be affected by the
inclusion of $u$ in the seed set.
\end{proof}

\paragraph*{Proof of Theorem 2}\mbox{}

First we establish the following lemma.
Let $r$ be the number of iterations executed by the repeat loop in Algorithm~\ref{algo:naivegreedy}. Let $S$ be the current seed set and $S^*$ be the optimal seed set. Without loss of generality, we may renumber nodes that was added to $S$ follow the chronicle order $S = \{u_1, u_2,\cdots, u_l\}$. Let $S_i = \bigcup_{j=1}^iu_j$ and let $j_i$ be the index of the iteration in which $u_i$ was considered.

\vspace{0.1in}
\begin{lem}
\label{lem1}
After each iteration $j_i, i = 1,\cdots, l+1$, the following holds:
\beq
\sigma(S_i) \geq \left[1-\prod_{k=1}^i{\left(1 - \frac{c(k)}{b}\right)}\right]\sigma(S^*).
\eeq
\end{lem}
\vspace{0.1in}

\begin{proof}
The proof of Lemma~\ref{lem1} was first presented by Khuller {\it et al.} in~\cite{Khuller199939} for the budgeted maximum coverage problem, which is a special case of BIM where all the active edge probabilities are 1. Later, it was extended by Krause {\it et al.} (Lemma 3 in~\cite{KrauseTechReport05}) for general submodular functions.
\end{proof}

Now we're in position to prove Theorem 2:

\begin{proof}
(Adapted from~\cite{Khuller199939}) We prove Theorem~\ref{thm2} by case analyzing Algorithm~\ref{algo:improvedgreedy}.
\begin{itemize}
\item {\bf Case 1:} If there exist at lease a node $u \in V$ which has spread greater than $\frac{1}{2}\sigma(S^*)$, then $u$ or any other nodes which possess a greater spread, will be selected as the second candidate $S_2$. Algorithm~\ref{algo:improvedgreedy} will therefore guarantee at least $\frac{1}{2}\sigma(S^*)$.
\item {\bf Case 2:} If there is no such node.
\begin{itemize}
\item \underline{{\it Case 2.1:}} If $c(S) < \frac{1}{2}b$, then we have $c(u) > \frac{1}{2}b, \forall u \not\in S$ since there is no more node that can be added to $S$ without violating the budget constrain. W.l.o.g, we assume $S \neq S^*$. Therefore, $S^*\backslash S$ contains at most 1 node $v$, otherwise $c(S^*) > b$. By submodularity definition we have,
    $$
    \hspace{-0.5in}
    \begin{array}{lll}
    \sigma(S^* \cap S) + \sigma(v) & \geq & \sigma((S^* \cap S) \cup v) + \sigma((S^* \cap S) \cap v)\\
                                    & \geq & \sigma(S^*) + \sigma(\emptyset)\\
                                    & \geq & \sigma(S^*).
    \end{array}
    $$
    By assumption, we have $\sigma(v) < \frac{1}{2}\sigma(S^*)$, therefore $\sigma(S^* \cap S) \geq \frac{1}{2}\sigma(S^*)$. It follows that $\sigma(S) \geq \frac{1}{2}\sigma(S^*)$.
\item \underline{{\it Case 2.2:}} If $c(S) \geq \frac{1}{2}b$. We first observe that for $a_1,\cdots a_n\in \mathbb{R}$ and $\sum_{i = 1}^{n}{a_i} \ge \alpha A$, the function,
    $$
    \hspace{-0.5in}
    \prod_{i=1}^{n}{\left(1-\frac{a_i}{A}\right)}
    $$
    is maximized when $a_i = \frac{\alpha A}{n}$. By Lemma~\ref{lem1}, we have,
    $$
    \hspace{-0.5in}
    \begin{array}{lll}
    \sigma(S_i) & \geq & \displaystyle \left[1 - \prod_{k=1}^{i}{\left(1-\frac{c(k)}{b}\right)}\right]\sigma(S^*)\\
                & \geq & \displaystyle \left[1 - \left(1-\frac{1}{2i}\right)^i\right]\sigma(S^*)\\
                & \geq & \displaystyle \left(1 - \frac{1}{\sqrt{e}}\right)\sigma(S^*).
    \end{array}
    $$
    Thus, in the worst case, Algorithm~\ref{algo:improvedgreedy} provides a $(1-1/\sqrt{e})$-approximation.
\end{itemize}
\end{itemize}
\end{proof}

%
%

\end{document}

%% file: Commands.tex
\usepackage{paralist}

\usepackage[ruled,vlined]{algorithm2e}

\usepackage{amssymb}
\usepackage{amsmath} 
\usepackage{verbatim}
\usepackage{array}
\usepackage{graphicx}
\usepackage{epstopdf}
\usepackage{stmaryrd}
\usepackage{dsfont}
\usepackage{subfig}
\usepackage{soul}

\usepackage{framed}
\usepackage[dvipsnames]{color}
\definecolor{shadecolor}{gray}{0.85}

\newif\iflong 
\longfalse

\newif\ifcomm
\commtrue

\newtheorem{thm}{Theorem}

\newtheorem{lem}{Lemma}
\newtheorem{prop}{Proposition}

\newtheorem{defn}{Definition}

\newcommand{\argmax}{\arg\max}

\newcommand{\Prob}[1]{{\mathbb P}\left(#1\right)}    
\newcommand{\E}{{\mathbb E}}                         

\newenvironment{proof-sketch}{{\noindent\em Proof Sketch.}\hspace*{0.3em}}{\qed\medskip}
\newenvironment{proof-of}[1]{{\noindent\em Proof of #1.}\hspace*{0.3em}}{\qed\medskip}

\newcounter{assumption}
\newcommand{\theassumptionletter}{A}
\renewcommand{\theassumption}{\theassumptionletter\arabic{assumption}}

\newcommand{\MG}{\mathcal{G}}

\newcommand{\MD}{\mathcal{D}}

\newcommand{\MP}{\mathcal{P}}

\newcommand{\beq}{\begin{equation}}
\newcommand{\eeq}{\end{equation}}
\newcommand{\beqa}{\begin{eqnarray}}
\newcommand{\eeqa}{\end{eqnarray}}
\newcommand{\beqan}{\begin{eqnarray*}}
\newcommand{\eeqan}{\end{eqnarray*}}
\newcommand{\ben}{\begin{eqnarray*}}
\newcommand{\een}{\end{eqnarray*}}

\ifcomm
   \newcommand\comm[1]{\textcolor{blue}{ #1}}
\else
   \newcommand\comm[1]{}
   \renewcommand{\todo}[1]{}
\fi

\newcommand{\ttt}{\mbox{$\mathbf{t}$}}
\newcommand{\xx}{\mbox{$\mathbf{x}$}}